\documentclass[12pt]{article}

\usepackage{graphicx,psfrag,epsf}
\usepackage{mathtools}
\usepackage{booktabs}
\usepackage[english]{babel} 
\usepackage{amsmath,amsfonts,amsthm}
\usepackage{amssymb}
\usepackage{graphicx}
\usepackage{array}
\usepackage{chngcntr}
\usepackage[toc,page]{appendix}
\usepackage{booktabs} 
\usepackage{natbib}
\usepackage[nodisplayskipstretch]{setspace}
\usepackage{hyperref}
\usepackage{floatrow}
\usepackage{subfigure}
\usepackage{subfloat}
\usepackage{caption}
\usepackage{accents}
\usepackage{epstopdf}
\usepackage{mathrsfs}
\usepackage{color}
\usepackage{enumitem}
\usepackage{tabularx}
\usepackage[font = small,labelfont=bf,textfont=it]{caption} 
\usepackage{footnote}
\usepackage{algorithm}
\usepackage[noend]{algpseudocode}
\usepackage{rotating}
\usepackage{tikz}
\usetikzlibrary{shapes.geometric, arrows}
\tikzstyle{red} = [rectangle, rounded corners, minimum height=1cm,text centered, text width=6cm, draw=black, fill=red!30]
\tikzstyle{gray} = [rectangle, rounded corners, minimum height=1cm,text centered, text width=6cm, draw=black, fill=gray!30]
\tikzstyle{arrow} = [thick,->,>=stealth]
\usepackage{footnote}
\usepackage{etoolbox}

\makeatletter
\interfootnotelinepenalty=\@M
\newcommand{\distas}[1]{\mathbin{\overset{#1}{\kern\z@\sim}}}%
\newcommand{\bm}[1]{\mathbf{#1}}
\newsavebox{\mybox}\newsavebox{\mysim}
\newcommand{\distras}[1]{%
  \savebox{\mybox}{\hbox{\kern3pt$\scriptstyle#1$\kern3pt}}%
  \savebox{\mysim}{\hbox{$\sim$}}%
  \mathbin{\overset{#1}{\kern\z@\resizebox{\wd\mybox}{\ht\mysim}{$\sim$}}}%
}
\newtheorem{theorem}{Theorem}
\newtheorem{definition}{Definition}

\newtheorem{lemma}{Lemma}

\newtheorem{proposition}{Proposition}
\newcolumntype{C}[1]{>{\centering\let\newline\\\arraybackslash\hspace{0pt}}m{#1}}

\newcommand{\be}{\begin{equation}}
\newcommand{\ee}{\end{equation}}
\newcommand{\bi}{\begin{itemize}}
\newcommand{\ei}{\end{itemize}}
\newcommand{\ben}{\begin{enumerate}}
\newcommand{\een}{\end{enumerate}}
\newcommand{\stb}{\State $\bullet$ \;}
\newcommand{\crd}[1]{{\color{red}{#1}}}
\newcommand{\cblb}[1]{{\color{blue}{\textbf{#1}}}}
\newcommand{\ubar}[1]{\underaccent{\bar}{#1}}
\allowdisplaybreaks
\DeclareMathOperator*{\argmin}{\arg\!\min}

\makeatother

\let\oldbibliography\thebibliography
\renewcommand{\thebibliography}[1]{\oldbibliography{#1}
\setlength{\itemsep}{0pt}} 

\newcommand{\blind}{1}

\patchcmd{\footnotemark}{\stepcounter{footnote}}{\refstepcounter{footnote}}{}{}
\newcounter{savecntr}
\newcounter{restorecntr}

\author{
\textsc{Simon Mak}, \textsc{V. Roshan Joseph}\\[2mm] 
\normalsize Georgia Institute of Technology \\ 
\vspace{-5mm}
}
\title{Projected support points: a new method for high-dimensional data reduction}

\begin{document}
\def\spacingset#1{\renewcommand{\baselinestretch}%
{#1}\small\normalsize} \spacingset{1}

\if1\blind
{
  \title{\bf Projected support points: a new method for high-dimensional data reduction}
  \small
  \author{Simon Mak\setcounter{savecntr}{\value{footnote}}\thanks{Stewart School of Industrial and Systems Engineering, Georgia Institute of Technology, Atlanta, GA}, \; V. Roshan Joseph\setcounter{restorecntr}{\value{footnote}}%
  \setcounter{footnote}{\value{savecntr}}\footnotemark
  \setcounter{footnote}{\value{restorecntr}}
\footnote{Corresponding author} 
\thanks{This work is supported by the U. S. Army Research Office grant W911NF-17-1-0007, and by the NSF DMS grant 1712642.}
}
  \maketitle
} \fi

\if0\blind
{
  \bigskip
  \bigskip
  \bigskip
  \begin{center}
    {\LARGE\bf Projected support points: a new method for high-dimensional data reduction}
\end{center}
  \medskip
} \fi

\bigskip

\begin{abstract}
In an era where big and high-dimensional data is readily available, data scientists are inevitably faced with the challenge of reducing this data for expensive downstream computation or analysis. To this end, we present here a new method for reducing high-dimensional big data into a representative point set, called projected support points (PSPs). A key ingredient in our method is the so-called sparsity-inducing (SpIn) kernel, which encourages the preservation of low-dimensional features when reducing high-dimensional data. We begin by introducing a unifying theoretical framework for data reduction, connecting PSPs with fundamental sampling principles from experimental design and Quasi-Monte Carlo. Through this framework, we then derive sparsity conditions under which the curse-of-dimensionality in data reduction can be lifted for our method. Next, we propose two algorithms for one-shot and sequential reduction via PSPs, both of which exploit big data subsampling and majorization-minimization for efficient optimization. Finally, we demonstrate the practical usefulness of PSPs in two real-world applications, the first for data reduction in kernel learning, and the second for reducing Markov Chain Monte Carlo (MCMC) chains.

\end{abstract}

\vfill

\newpage
\spacingset{1.45} 

\section{Introduction}

In an era with remarkable advancements in computer engineering, computational algorithms and mathematical modeling, statisticians and data scientists are inevitably faced with the challenge of working with \textit{big} and \textit{high-dimensional} data. For many such applications, data reduction -- the reduction of big data (assumed here to be on $\mathbb{R}^p$) to a smaller, representative dataset -- is a necessary first step. The reason for this is two-fold. First, the interest often lies not in the big data itself, but the propagation of this data via some downstream computation, which we denote by function $g: \mathbb{R}^p \rightarrow \mathbb{R}$. When $g$ is \textit{costly} to evaluate (e.g., time-consuming or expensive), this propagation can be performed only for a small fraction of the big data. Second, the manipulation of big data (e.g., for inference or prediction) can demand massive \textit{memory} and \textit{storage} costs. Such costs often exceed the resources available in standard computers, and data reduction is a necessary step to achieve any analysis. Both problems are further compounded in high dimensions (i.e., $p \gg 1$), and careful analysis of $g$ is needed to understand the specific sparsity structure needed to ensure a meaningful reduction. To this end, we present a novel methodology for reducing high-dimensional data into a representative dataset, called \textit{projected support points} (PSPs), which preserves low-dimensional attributes via a new kernel function.

With the increasing prevalence of big and high-dimensional data, our reduction method can be used in a broad range of real-world statistical applications. One such application is for speeding up kernel methods \citep{Fea2001} in statistical learning, which are computationally expensive for large training datasets \citep{RW2006}. By reducing big data into a smaller, representative dataset, the proposed PSPs can allow for effective learning given a computation budget. Another important application is for reducing Markov Chain Monte Carlo (MCMC) samples in Bayesian computation \citep{Gea1995}. Our reduction method can be particularly effective for Bayesian modeling of engineering problems, where high-dimensional parameters from MCMC chains often need to be pushed forward via expensive simulations \citep{Mea2017}. Other applications, among many, include computer experiment design \citep{Sea2013}, uncertainty quantification \citep{Smi2013}, and scenario reduction in stochastic programming \citep{Dea2003}.


The problem of data reduction is an active area of research among statisticians and computer scientists, and much progress has been made in recent years. Some notable work (among many) include \cite{HK2005}, \cite{Fea2011} and \cite{Hea2016}, who proposed reduction methods for $k$-means clustering, Gaussian mixture model fitting, and Bayesian logistic regression, respectively. Another popular approach to data reduction is via \textit{mean-matching} (see \citealp{Gea2009}) -- the idea is to have the reduced data well-approximate the full data, by ensuring their downstream sample means are close over a class of ``reasonable'' downstream maps $g$. Recent work on this includes (a) \textit{kernel herding} \citep{Cea2012}, which using a kernel $\gamma$, generates a point sequence to successively match sample means over a function space for $g$, and (b) support points (SPs, \citealp{MJ2016b}), which employ parallelized convex programming and pairwise distances to efficiently perform reduction. One disadvantage of these two methods is that it assumes the downstream map $g$ is active in \textit{all} $p$ variables; in high dimensions ($p \gg 1$), it is much more likely that $g$ is active for only a \textit{small fraction} of these variables. By neglecting low-dimensional structure in $g$, existing reduction methods can experience a so-called \textit{curse-of-dimensionality}, in that they yield poor reduction of high-dimensional data. Our methodology addresses this, by (a) introducing a new \textit{Sparsity-Inducing} (SpIn) kernel for targeting low-dimensional features in reduction, and (b) investigating the sparsity structure on $g$ needed to theoretically lift this curse-of-dimensionality.

The notion of low-dimensional structure in high-dimensional functions has been explored in both the numerical integration (Quasi-Monte Carlo, or QMC) and experimental design communities. In QMC, this began with the idea of \textit{effective dimension} \citep{Cea1997,SW1998}, which quantified the belief that certain dimensions of an integrand are more important than others. In recent years, this has culminated into a body of work investigating the \textit{tractability} conditions of integration on the uniform unit hypercube \citep{KS2005, Dea2013}. Simply put, these conditions provide the sparsity structure needed to lift the curse-of-dimensionality for high-dimensional integration \citep{NW2008}. This attention to low-dimensional structure is mirrored in experimental design. Indeed, the principles of effect sparsity, hierarchy and heredity \citep{BH1961, HW1992, WH2011} -- fundamental principles for designing and analyzing experimental data -- can be seen as low-dimensional guiding rules for learning high-dimensional functions. Recently, these principles were further developed in \cite{Jea2015}, who proposed a new experimental design with good space-filling properties on projections of a uniform design space. The above literature, however, investigates sampling strategies for functions with an underlying \textit{uniform} measure on its domain, and are therefore not directly applicable for big data reduction (since data is almost \textit{never} uniformly distributed in practice). To this end, our approach provides a unifying framework \textit{extending} this body of work for data reduction, from which novel theoretical insights and practical algorithms can be derived for high-dimensional reduction.

This paper is organized as follows. Section \ref{sec:prob} presents the data reduction framework, and introduces the proposed PSP method. Section \ref{sec:theory} investigates its theoretical properties, and establishes sparsity conditions for lifting the curse-of-dimensionality. Section \ref{sec:bayesian} provides guidelines on the specification of the SpIn kernel. Section \ref{sec:algo} presents two algorithms for one-shot and sequential reduction. Section \ref{sec:app} demonstrates the effectiveness of PSPs in simulations and in two real-world applications. Finally, Section \ref{sec:concl} concludes with thoughts on future work. For brevity, all proofs of technical results are deferred to the Appendix.

\section{Problem framework}
\label{sec:prob}

\begin{figure}[t]
\centering
\includegraphics[width=0.85\textwidth]{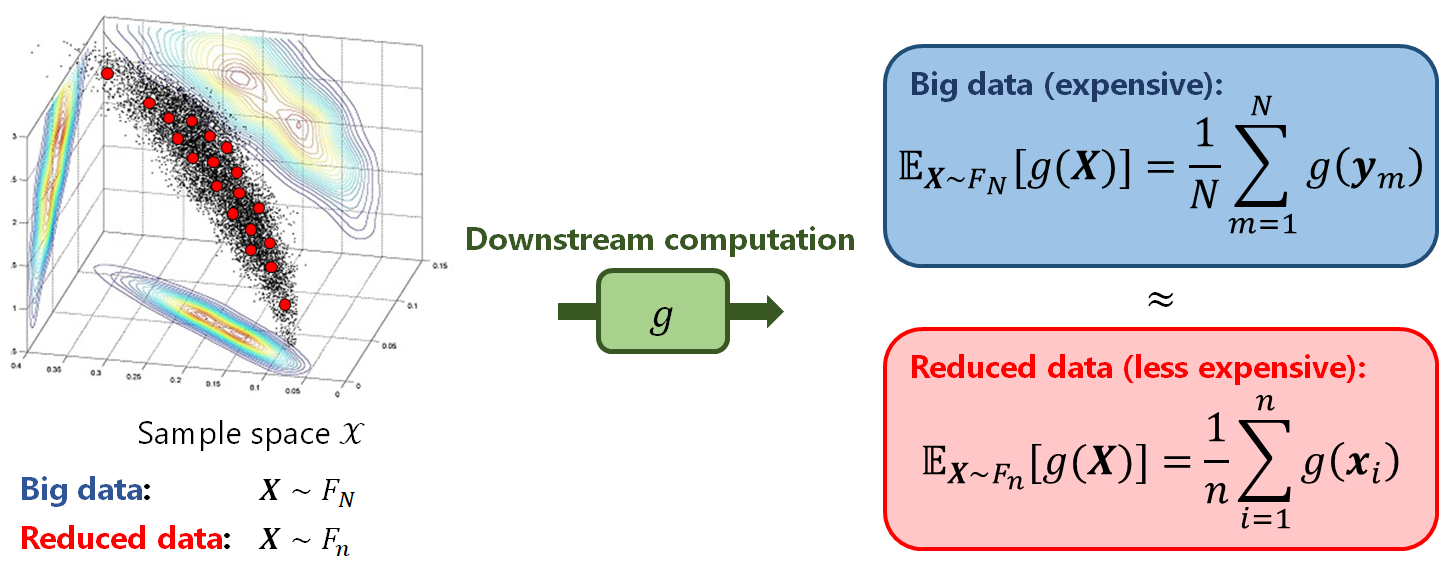}
\caption{A visualization of the considered framework for data reduction.}
\label{fig:framework}
\end{figure}

In this section, we first introduce the considered data reduction framework and the proposed PSP method, then illustrate via a motivating example why PSPs can improve upon existing methods for high-dim. reduction.

\subsection{Data reduction}

We begin by highlighting in Figure \ref{fig:framework} the three components of our data reduction framework. First, on the left, the black points represent the big data $\{\bm{y}_m\}_{m=1}^N \subseteq \mathcal{X} \subseteq \mathbb{R}^p$ to be reduced, with empirical distribution (e.d.f.) $F_N$, where $N \gg 1$ is the size of the big data. The red points show one reduction of this big data to a smaller point set $\mathcal{D} := \{\bm{x}_i\}_{i=1}^n$, where $n \ll N$ is the reduced sample size. Next, in the middle, the green arrow shows the downstream computation $g: \mathcal{X} \rightarrow \mathbb{R}$. The implicit assumption here is that $g$ is black-box and \textit{costly} to evaluate, otherwise there would be no need for data reduction. Given the costly nature of $g$, we assume the sample size $n$ is set to maximize the number of evaluations given a practical cost budget (e.g., on computational resources). Finally, on the right, our goal is to perform reduction in such a way that the ``reduced'' sample mean $\mathbb{E}_{\bm{X} \sim F_n}[ g(\bm{X}) ]$ well-approximates the ``full'' sample mean $\mathbb{E}_{\bm{X} \sim F_N}[ g(\bm{X}) ]$, over a large class of ``reasonable'' functions for $g$. This mean-matching approach is applicable to a wide range of statistical applications (see Section \ref{sec:app}).

As typical in machine learning (see, e.g., \citealp{Bot2010}), we assume that big data is drawn from an underlying distribution function (d.f.) $F$; this allows for more amenable analysis and clearer exposition of ideas. In particular, in order to derive meaningful theoretical insights in the next two sections, we adopt the limiting view that the \textit{distribution} $F$ is the big data to reduce; this is akin to having an \textit{infinite} amount of data. We return to the finite view of big data when presenting the proposed algorithms in Section \ref{sec:algo}. 

Of course, in practice the downstream computation $g$ is never known. One solution (which we adopt) is to first assume $g$ comes from a space (call this $\mathcal{H}$) of ``reasonable'' functions, then perform data reduction by ensuring the mean-matching goal is satisfied for \textit{all} $g \in \mathcal{H}$. We briefly review two important ingredients: the reproducing kernel Hilbert space and the kernel discrepancy, then show how these can be connected for data reduction.

The first ingredient is the reproducing kernel Hilbert space (RKHS), defined below:
\begin{definition}[RKHS; \citealp{Aro1950}]
Let $\gamma : \mathcal{X} \times \mathcal{X} \rightarrow \mathbb{R}$ be a symmetric, positive-definite (p.d.) kernel. The \textup{reproducing kernel Hilbert space (RKHS)} $(\mathcal{H}_\gamma, \langle \cdot,\cdot \rangle_{\gamma})$ for kernel $\gamma$ is comprised of the function space:
\begin{equation}
\mathcal{H}_{\gamma} := {\textup{span}\{ \gamma(\cdot,\bm{x}) : \bm{x} \in \mathcal{X} \}},
\label{eq:rkhs}
\end{equation}
endowed with the inner product:
\begin{equation}
\langle f , g \rangle_{\gamma} := \sum_{j=1}^s \sum_{j'=1}^{s'} \alpha_j \beta_{j'} \gamma(\bm{x}_j, \bm{x}_{j'}), \; f(\bm{x}) = \sum_{j=1}^s \alpha_j \gamma(\bm{x},\bm{x}_j), \; g(\bm{x}) = \sum_{j'=1}^{s'} \beta_{j'} \gamma(\bm{x},\bm{x}_{j'}).
\label{eq:inprod}
\end{equation}
\vspace{-0.2cm}
\label{def:rkhs}
\end{definition}
\noindent In words, given kernel $\gamma$, its RKHS $\mathcal{H}_\gamma$ can be constructed by taking the span of the reproducing kernel feature map $\gamma(\cdot, \bm{x})$, over domain $\mathcal{X}$. We will use the RKHS to model the space of ``reasonable'' downstream maps $g$, as it offers nice theoretical properties.

The second ingredient is the kernel discrepancy, defined below:
\begin{definition}[Kernel discrepancy; \citealp{Hic1998}]
Let $F$ be a d.f. on $\mathcal{X} \subseteq \mathbb{R}^p$, and let $F_n$ be the e.d.f. of a point set $\{\bm{x}_i\}_{i=1}^n \subseteq \mathcal{X}$. For a symmetric, p.d. kernel $\gamma$, the \textup{kernel discrepancy} between $F$ and $F_n$ is defined as:
\begin{equation}
D_{\gamma}(F,F_n) := \sqrt{\int_{\mathcal{X}} \int_{\mathcal{X}} \gamma(\bm{x},\bm{y})\; d[F-F_n](\bm{x}) \; d[F-F_n](\bm{y}) }.
\label{eq:kerdis}
\end{equation}
\vspace{-0.4cm}
\label{def:kerdis}
\end{definition}
\noindent In words, the kernel discrepancy $D_{\gamma}(F,F_n)$ measures how different $F$ and $F_n$ are, by weighing the difference in measure $F-F_n$ with kernel $\gamma$. A larger discrepancy suggests the point set $\{\bm{x}_i\}_{i=1}^n$ differs greatly from $F$, whereas a smaller discrepancy suggests the point set is quite similar to $F$. This kernel discrepancy is also known as the \textit{maximum mean discrepancy} in the machine learning literature, where it has been successfully applied for goodness-of-fit testing and neural network training (see, e.g., \citealp{Gea2012}).

These two ingredients can then be linked via the following upper bound:
\begin{lemma}[Koksma-Hlawka; \citealp{Hic1998}]
Let $\gamma$ be a symmetric, p.d. kernel on $\mathcal{X}$, and let $F$ and $F_n$ be as in Definition \ref{def:kerdis}. With $F_n$ approximating $F$, the \textup{integration error} of $g \in \mathcal{H}_\gamma$, defined as:
\begin{equation}
I(g;F,F_n) := \left| \int_{\mathcal{X}} g(\bm{x}) dF(\bm{x}) - \frac{1}{n}\sum_{i=1}^n g(\bm{x}_i) \right|,
\label{eq:interr}
\end{equation}
can be \textup{uniformly bounded} as:
\begin{equation}
\sup_{g \in \mathcal{H}_\gamma, \|g\|_{\gamma} \leq 1} I(g;F,F_n) = D_{\gamma}(F,F_n).
\label{eq:kh}
\end{equation}
\vspace{-0.8cm}
\label{lem:kh}
\end{lemma}
\noindent This lemma shows, for $g$ in the unit ball $\mathcal{B}_\gamma = \{g: \|g\|_{\gamma} \leq 1 \}$, the worst-case integration error for a point set $\{\bm{x}_i\}_{i=1}^n$ approximating $F$ is precisely the kernel discrepancy $D_{\gamma}(F,F_n)$. Viewed another way, by finding a reduction $\{\bm{x}_i\}_{i=1}^n$ which well-matches the big data $F$ by minimizing discrepancy $D_{\gamma}(F,F_n)$, we ensure its downstream computation also well-matches that for the full data, over all possible downstream computations $g \in \mathcal{B}_\gamma$. Recall from the Introduction that, for high-dim. problems (i.e., $p$ large), the downstream map $g$ is typically active for only a small subset of the $p$ variables. Our strategy is to incorporate this low-dim. prior belief into a new kernel for $\gamma$, called the \textit{sparsity-inducing kernel}, so that the unit ball $\mathcal{B}_\gamma$ consists of high-dim. functions with this desired low-dim. structure.


\subsection{Existing reduction methods}
\label{sec:exi}
Before presenting this new kernel, we provide a brief overview of two related methods in the literature: kernel herding \citep{Cea2012}, and support points \citep{MJ2016b}. First, given a symmetric p.d. kernel $\gamma$, kernel herding (or simply herding) generates the reduced dataset $\{\bm{x}_i\}_{i=1}^n$ via the sequential optimization scheme: 
\begin{equation}
\bm{x}_{n+1} \leftarrow \underset{\bm{x} \in \mathcal{X}}{\textup{Argmax}} \left\{ \mathbb{E}_{\bm{Y} \sim F}[\gamma(\bm{x},\bm{Y})] - \frac{1}{n+1} \sum_{i=1}^n \gamma(\bm{x},\bm{x}_i) \right\}.
\label{eq:herd}
\end{equation}
Here, $\mathbb{E}_{\bm{Y} \sim F}[\gamma(\bm{x},\bm{Y})]$ is typically approximated by the (finite) big data mean $\mathbb{E}_{\bm{Y} \sim F_N}[\gamma(\bm{x},\bm{Y})]$. By expanding \eqref{eq:kerdis}, the above scheme can be shown to be a greedy, point-by-point minimization of discrepancy $D_{\gamma}(F,F_n)$. In practice, standard kernels are used for $\gamma$, the most popular being the standard Gaussian kernel $\exp\{ -\sum_{l=1}^p (x_l - y_l)^2 \}$. For simplicity, we refer to herding as the point sequence in \eqref{eq:herd} with $\gamma$ as the standard Gaussian kernel.

The support points (SPs) in \cite{MJ2016b} also minimize the discrepancy $D_{\gamma}(F,F_n)$, with kernel $\gamma(\bm{x},\bm{y}) = -\|\bm{x}-\bm{y}\|_2$. The resulting discrepancy with this distance kernel is known as the \textit{energy distance} \citep{SZ2013}, a popular non-parametric test statistic for goodness-of-fit. \cite{MJ2016b} showed that SPs converge in distribution to the desired measure $F$, and enjoy improved performance (in terms of integration rate) over Monte Carlo sampling over a large function class. This distance kernel also allows for efficient data reduction via parallelized difference-of-convex programming.

\begin{figure}
\centering
\includegraphics[width=\textwidth]{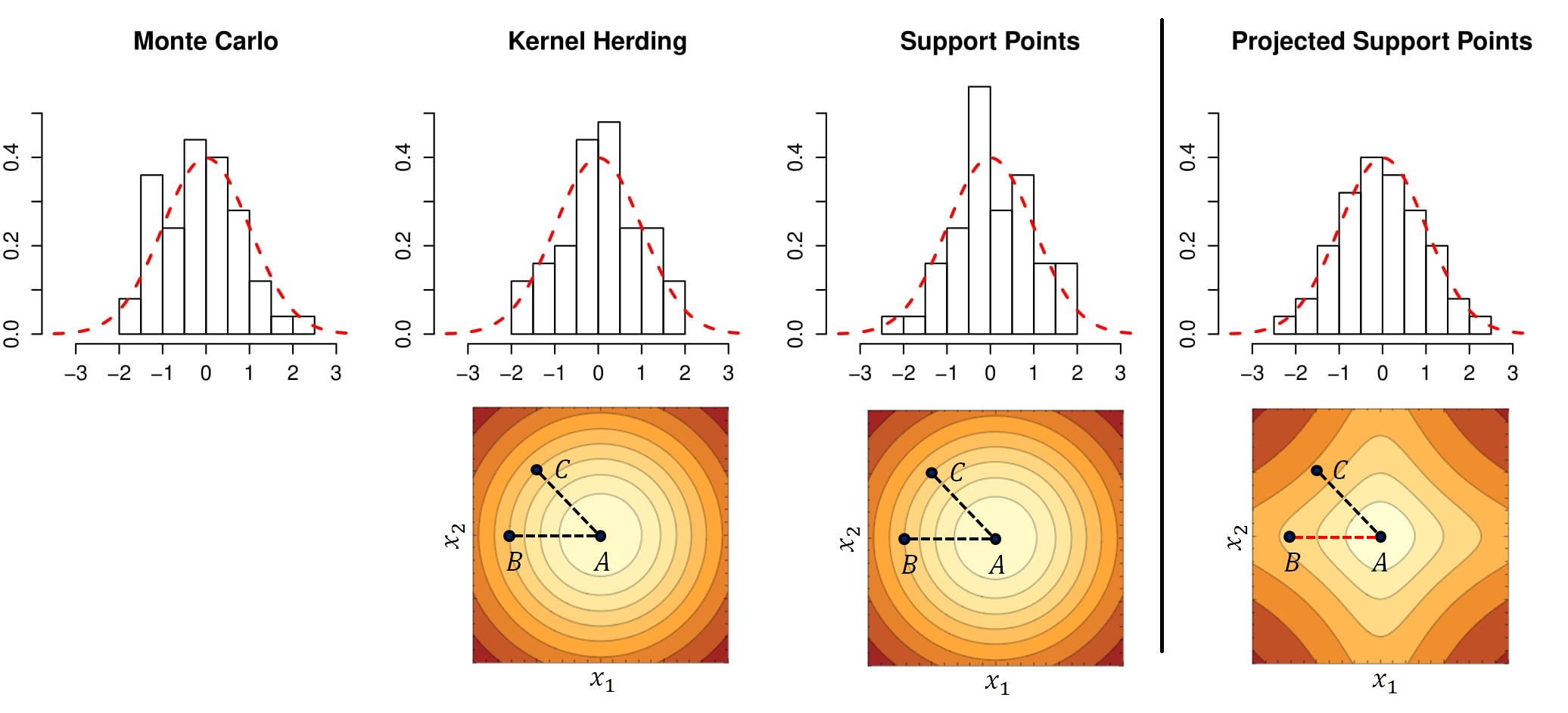}
\caption{(Top) 1-d projections of the reduced $n=50$ point sets for $F = $ 10-d i.i.d. $\mathcal{N}(0,1)$. (Bottom) Contours of the kernel $\gamma$ used for reduction.}
\label{fig:projill}
\end{figure}

For the desired goal of data reduction, however, both herding and SPs have a key disadvantage: they can yield poor reduction of low-dim. features in high-dim. data. To see this, take the following example. Suppose big data is generated from $F = $ 10-d i.i.d. $\mathcal{N}(0,1)$ distribution, and consider its reduction into $n=50$ points. The three plots in Figure \ref{fig:projill} (top left) show the 1-d projections of the reduced point sets from Monte Carlo, herding (using the std. Gaussian kernel) and SPs, with the true 1-d marginal density in red. Note that, while herding and SPs provide an optimized reduction of $F$ in the \textit{full} 10-d space, both methods give a poor reduction (even worse than Monte Carlo!) of the 1-d marginal distribution. For downstream maps $g$ depending on only this one variable (in general, $g$ with low-dim. structure), Lemma \ref{lem:kh} suggests these optimized methods can perform \textit{worse} than random sampling.

One reason for the poor performance of herding and SPs in high-dimensions (a so-called \textit{curse-of-dimensionality}) is the choice of kernel $\gamma$, namely, the std. Gaussian and $-\| \cdot\|_2$ kernels. Figure \ref{fig:projill} (bottom left) shows the contours for these two kernels. Viewing kernel $\gamma$ as a similarity measure, with larger $\gamma(\bm{x},\bm{y})$ indicating greater similarity between points $\bm{x}$ and $\bm{y}$, note that points $(A,B)$ and $(A,C)$ are assigned the same similarity by both kernels, since the Euclidean distance is the same for both point pairs. However, suppose one knows $g$ is active in only one of the two variables, say, $x_2$ (in general, in a low-dim. projection). In this case, one should assign \textit{greater} similarity to $(A,B)$, since these points have the \textit{same} $x_2$ coordinate. We propose below a new kernel which captures this desired low-dim. structure.

\subsection{SpIn kernel and PSPs}
For the sparsity-inducing kernel, we begin with the general Gaussian kernel:
\begin{equation}
\gamma_{\boldsymbol{\theta}}(\bm{x},\bm{y}) := \exp\left\{ - \sum_{\varnothing \neq \bm{u} \subseteq [p]} \theta_{\bm{u}} \|\bm{x}_{\bm{u}}-\bm{y}_{\bm{u}}\|_2^2 \right\}, \quad [p] := \{1, \cdots, p\},
\label{eq:gamma}
\end{equation}
with scale parameters $\boldsymbol{\theta} = (\theta_{\bm{u}})_{|\bm{u}|=1}^p \geq 0$ following the so-called \textit{product-and-order} (POD; \citealp{Kea2012}) form:
\begin{equation}
\theta_{\bm{u}} = \Gamma_{|\bm{u}|}^{(\theta)} \prod_{l \in \bm{u}} \theta_l.
\label{eq:podtheta}
\end{equation}
The key intuition here is that a larger scale parameter $\theta_{\bm{u}}$ indicates a greater importance of subspace $\bm{u} \subseteq [p]$. As we show in later sections, this provides a flexible framework for encoding the desired sparsity structure for data reduction. We do not provide here a full justification for the POD form in \eqref{eq:podtheta} (this is given later in Sections \ref{sec:theory} and \ref{sec:bayesian}), other than to mention that the \textit{product} weights $(\theta_l)_{l=1}^p$ quantify variable importance, and the \textit{order} weights $(\Gamma_{|\bm{u}|}^{(\theta)})_{|\bm{u}|=1}^\infty$ quantify order importance. We also note that, while the Gaussian kernel in \eqref{eq:gamma} allows for insightful theoretical analysis in Section \ref{sec:theory}, our method can be extended for any scale-parametrized kernel in practice.

Prior to observing data on the black-box $g$, one typically has no information on which variables are important and which are not. In high dimensions, however, we do know that $g$ is likely to be \textit{sparse}, in that it is active for only \textit{some} of the $p$ variables. One way to incorporate this sparsity within kernel $\gamma_{\boldsymbol{\theta}}(\bm{x},\bm{y})$ in \eqref{eq:gamma} is to assume a prior distribution $\pi$ on product weights $(\theta_l)_{l=1}^p$, which quantify variable importance. From intuition, this prior should assign high probability to $\boldsymbol{\theta}$ with large values in a small subset of its entries. Given such a prior $\pi$, the \textit{sparsity-inducing} (SpIn) kernel is defined as follows:
\begin{definition}[SpIn kernel]
Let $\pi$ be a (proper) prior on product weights $(\theta_l)_{l=1}^p$, and suppose the order weights $(\Gamma_{|\bm{u}|}^{(\theta)})_{|\bm{u}|=1}^\infty$ are fixed. The \textup{sparsity-inducing (SpIn)} kernel under prior $\pi$ is:
\begin{equation}
\gamma_{\boldsymbol{\theta} \sim \pi}(\bm{x},\bm{y}) := \mathbb{E}_{\boldsymbol{\theta} \sim \pi} [ \gamma_{\boldsymbol{\theta}}(\bm{x},\bm{y}) ].
\end{equation}
\end{definition}
\noindent In words, the SpIn kernel $\gamma_{\boldsymbol{\theta} \sim \pi}$ can be seen as an \textit{averaged} similarity measure between two points, under the prior assumption (from $\pi$) that only a subset of variables are important. 

We can now define the proposed \textit{projected support points} (PSPs) for high-dim. reduction:
\begin{definition}[PSPs]
Let $F$ be a d.f. on $\mathcal{X} \subseteq \mathbb{R}^p$. 
\bi
\item Suppose the weights $\boldsymbol{\theta} = (\theta_{\bm{u}})_{|\bm{u}|=1}^p$ are fixed. Then the \textup{$\boldsymbol{\theta}$-weighted PSPs} of $F$ are:
\begin{equation}
\argmin_{\bm{x}_1, \cdots, \bm{x}_n} D_{\gamma_{\boldsymbol{\theta}}}(F,F_n).
\label{eq:thetapsp}
\end{equation}
\item Suppose $\boldsymbol{\theta}$ follows a (proper) prior $\pi$, and let $\gamma_{\boldsymbol{\theta} \sim \pi}$ be the SpIn kernel under $\pi$. The \textup{$\pi$-expected PSPs} of $F$ are:
\begin{equation}
\argmin_{\bm{x}_1, \cdots, \bm{x}_n} D_{\gamma_{\boldsymbol{\theta} \sim \pi}}(F,F_n).
\label{eq:pipsp}
\end{equation}
\ei
\end{definition}
\noindent In words, the $\pi$-expected (or $\boldsymbol{\theta}$-weighted) PSPs minimize the discrepancy $D_{\gamma}(F,F_n)$ with the SpIn kernel $\gamma = \gamma_{\boldsymbol{\theta} \sim \pi}$ (or the $\boldsymbol{\theta}$-weighted kernel $\gamma = \gamma_{\boldsymbol{\theta}}$). The $\boldsymbol{\theta}$-weighted PSPs will be used in Section \ref{sec:theory} for theoretical analysis, while the $\pi$-expected PSPs will be used in practice for data reduction. The rationale for PSPs is that, by minimizing discrepancy with the SpIn kernel, the resulting reduction gives a better reduction of low-dim. features in high-dim. data. By Lemma \ref{lem:kh}, this then yields improved estimation of low-dim. downstream quantities, which is the desired goal.

To illustrate this intuition, consider again the earlier example of reducing big data from $F = $ 10-d i.i.d. $\mathcal{N}(0,1)$ distribution into $n=50$ points. Here, we assume a simple form of the SpIn kernel, with (a) $\theta_{l} \distas{i.i.d.} \text{Gamma}(0.1,0.01)$, and (b) $\Gamma_{1}^{(\theta)} = 1$ and $\Gamma_{k}^{(\theta)} = 0$ for $k > 1$. This reduces to an anisotropic Gaussian kernel, averaged over i.i.d. $\text{Gamma}(0.1,0.01)$ priors on scale parameters. Note that these i.i.d. priors provide one way of quantifying sparsity, since only a small number of product weights $(\theta_{l} )_{l=1}^p$ will be large with high probability. Figure \ref{fig:projill} (top right) shows the 1-d projections of the $n=50$ PSPs. We see that PSPs enjoy a noticeable improvement over the other three methods, yielding a near-perfect reduction of the true marginal distribution (in red) using only $n=50$ points. For downstream computations $g$ with low-dim. structure, PSPs can therefore offer improved performance over both Monte Carlo and existing reduction methods.

An inspection of the SpIn kernel contours (Figure \ref{fig:projill}, bottom right) shows why this is the case. Recall for the earlier std. Gaussian and distance kernels, points $(A,B)$ and $(A,C)$ are assigned the same similarity, despite the former having the same $x_2$ coordinate and the latter having different $x_2$ coordinates. The SpIn kernel, on the other hand, accounts for the fact that $(A,B)$ are close in a projected subspace (the $x_2$-axis), by assigning greater similarity to $(A,B)$ than $(A,C)$. By using a kernel which factors in low-dim. similarities, PSPs can yield an effective reduction of low-dim. features in data, as in this toy example.

With this, we now explore the low-dim. structure on $g$ imposed by this new kernel, to better understand the sparsity conditions required for effective, high-dim. data reduction.

\section{Theoretical analysis}
\label{sec:theory}

In this section, we first provide a brief summary of two well-established views on low-dim. structure in functions:  (a) the three effect principles from experimental design \citep{WH2011}, and (b) the notion of tractability from QMC \citep{NW2008, Dea2013}. Using these two views, we then derive the RKHS of $\gamma_{\boldsymbol{\theta}}$, and give some insight on the sparsity structure on $g$ needed for effective high-dim. reduction via PSPs. To allow for meaningful theoretical analysis, we assume in this section a simpler, anisotropic form for the kernel $\gamma_{\boldsymbol{\theta}}$ from \eqref{eq:gamma}:
\begin{equation}
\gamma_{\boldsymbol{\theta}} = \exp\left\{ - \sum_{l=1}^p \theta_l ({x}_l-{y}_l)^2 \right\},
\label{eq:aniso}
\end{equation}
with the full form in \eqref{eq:gamma} considered in later sections.

\subsection{Three effect principles and tractability}

In experimental design, a key challenge is learning functions in high-dimensions, using limited data from expensive experiments. Over the years, three fundamental principles (see \citealp{WH2011}) have been successfully applied, all of which exploit low-dim. structure on $g$. The first principle, called \textit{effect sparsity}, states that $g$ is likely comprised of a \textit{small} number of important effects. For example, in a function with $p=10$ variables, it is likely that only a subset of these 10 variables are truly active. The second, called \textit{effect hierarchy}, states that lower-order effects are more likely active than higher-order effects. For example, the main effect of variable $x_1$ (a first-order effect) is more likely active than the interaction effect of $x_1$ and $x_2$ (a second-order effect). The last principle, called (strong) \textit{effect heredity}, states that higher-order effects are active \textit{only when} all lower-order components are active. For example, an interaction effect of $x_1$ and $x_2$ is active only when the main effects of $x_1$ and $x_2$ are also active. Together, these principles give a flexible framework for learning low-dim. structure in high-dim. functions.

This attention to low-dim. structure has been mirrored in QMC, beginning with the idea of \textit{effective dimension} \citep{Cea1997,SW1998}: the belief that certain variables in an integrand $g$ are more important than others. This motivated a recent body of work investigating the \textit{tractability} of the integration problem on the uniform unit hypercube \citep{NW2008}, i.e., how difficult integration becomes as dimension $p$ increases. This study of tractability is important, as it provides insight on what sparsity structure is needed on $g$ to break the \textit{curse-of-dimensionality} for high-dim. integration. Of particular interest to us is the work on a \textit{dimension-free} integration rate \citep{KS2005, Dea2013}, which investigated sparsity conditions on $g$ to achieve an integration error rate which does not depend on dimension $p$. This dimension-free rate can be viewed as a \textit{strong} condition on tractability, since it requires the problem to not grow in difficulty as dimension $p$ increases.

Our theoretical analysis below makes use of both views to better understand the tractability of the proposed data reduction framework. We will first construct the RKHS of $\gamma_{\boldsymbol{\theta}}$, then derive the low-dim. structure on $g$ (via the three effect principles) required for achieving a dimension-free error rate $I(g;F,F_n)$ for data reduction.


\subsection{Dimension-free error rate}
We first give an explicit construction of the RKHS for kernel $\gamma_{\boldsymbol{\theta}}$ in \eqref{eq:aniso}:

\begin{theorem}[RKHS of $\gamma_{\boldsymbol{\theta}}$]
Let $(\mathcal{H}_{\gamma_{\boldsymbol{\theta}}}, \langle \cdot, \cdot \rangle_{\gamma_{\boldsymbol{\theta}}})$ be the RKHS for kernel $\gamma_{\boldsymbol{\theta}}$ in \eqref{eq:aniso}. Then:
\begin{equation}
\mathcal{H}_{\gamma_{\boldsymbol{\theta}}} = \left\{g: \mathbb{R}^p \rightarrow \mathbb{R} \; \Bigg| \; g(\bm{x}) = \exp(-\|\bm{x}\|_{\boldsymbol{\theta}}^2)\sum_{|{\boldsymbol{\alpha}}|=0}^\infty w_{\boldsymbol{\alpha}} \bm{x}^{\boldsymbol{\alpha}}, \; \|g\|_{\gamma_{\boldsymbol{\theta}}} < \infty\right\},
\label{eq:rkhs}
\end{equation}
with inner product given by:
\begin{equation}
\langle f, g \rangle_{\gamma_{\boldsymbol{\theta}}} = \sum_{k=0}^\infty \frac{k!}{2^k} \sum_{|{\boldsymbol{\alpha}}|=k} \frac{v_{\boldsymbol{\alpha}} w_{\boldsymbol{\alpha}}}{C_{\boldsymbol{\alpha}}^k {\boldsymbol{\theta}}^{\boldsymbol{\alpha}}}, \quad f(\bm{x}) = \exp(-\|\bm{x}\|_\theta^2)\sum_{|{\boldsymbol{\alpha}}|=0}^\infty v_{\boldsymbol{\alpha}} \bm{x}^{\boldsymbol{\alpha}}.
\label{eq:inner}
\end{equation}
Here, $\boldsymbol{\alpha}=(\alpha_1, \cdots, \alpha_p)$ with $|\boldsymbol{\alpha}|=\sum_{l=1}^p \alpha_l$, $\{w_{\boldsymbol{\alpha}}\}_{|\boldsymbol{\alpha}|=0}^\infty, \{v_{\boldsymbol{\alpha}}\}_{|\boldsymbol{\alpha}|=0}^\infty \subseteq \mathbb{R}$ are coefficients, $\bm{x}^{\boldsymbol{\alpha}} = \prod_{l=1}^p x_l^{\alpha_l}$ (similarly for $\boldsymbol{\theta}^{\boldsymbol{\alpha}}$) and $C_{\boldsymbol{\alpha}}^k = k!/({{\alpha}}_1! \cdots {{\alpha}}_p!)$ is the multinomial coefficient.
\label{thm:rkhs}
\end{theorem}
\noindent In other words, the RKHS $\mathcal{H}_{\gamma_{\boldsymbol{\theta}}}$ consists of functions spanned by the orthogonal basis functions $\{\exp(-\|\bm{x}\|_{\boldsymbol{\theta}}^2)\bm{x}^{\boldsymbol{\alpha}}\}_{|\boldsymbol{\alpha}|=0}^\infty$ -- the collection of monomials $\bm{x}^{\boldsymbol{\alpha}}$ multiplied by an exponential decay term $\exp(-\|\bm{x}\|_{\boldsymbol{\theta}}^2)$. The coefficients $\{w_{\boldsymbol{\alpha}}\}_{|\boldsymbol{\alpha}|=0}^\infty$ can be seen as ``ANOVA-like'' coefficients, quantifying the importance of each basis in $\{\exp(-\|\bm{x}\|_{\boldsymbol{\theta}}^2)\bm{x}^{\boldsymbol{\alpha}}\}_{|\boldsymbol{\alpha}|=0}^\infty$ for $g \in \mathcal{H}_{\gamma_{\boldsymbol{\theta}}}$; a larger ANOVA coefficient $w_{\boldsymbol{\alpha}}$ suggests greater importance of $\exp(-\|\bm{x}\|_{\boldsymbol{\theta}}^2)\bm{x}^{\boldsymbol{\alpha}}$ in $g$, and vice versa. This intuitive decomposition of $\mathcal{H}_{\gamma_{\boldsymbol{\theta}}}$ into ANOVA-like effects allows us to impose an interpretable low-dim. structure via the three effect principles.

Consider next a POD-like parametrization \citep{Kea2012} of the coefficients in \eqref{eq:rkhs}:
\begin{equation}
w_{\boldsymbol{\alpha}} = \mathrm{T}_{|\boldsymbol{\alpha}|}^{(w)}\prod_{l=1}^p w_l^{\alpha_l}.
\label{eq:pod}
\end{equation}
Here, the \textit{product} weights $(w_l)_{l=1}^\infty$ quantify the importance of each variable $x_l$ in $g$, with a larger value of $w_l$ suggesting greater importance for variable $x_l$. Similarly, the \textit{order} weights $(\mathrm{T}_{|\boldsymbol{\alpha}|}^{(w)})_{|\boldsymbol{\alpha}|=1}^\infty$ quantify the importance of effects of different orders in $g$, with a larger value of $\mathrm{T}_k^{(w)}$ indicating greater importance of $k$-th order effects. The key appeal of the parametrization in \eqref{eq:pod} is that it gives an intuitive way to quantify the three effect principles. In particular, \textit{effect sparsity} -- the belief that $g$ is comprised of a small number of effects -- can be imposed by enforcing a bounded condition on either product or order weights, thereby restricting the number of active effects in high dimensions. Similarly, \textit{effect hierarchy} -- the belief that lower-order effects dominate higher-order ones -- can be imposed by setting a decreasing sequence for order weights $(\mathrm{T}_{|\boldsymbol{\alpha}|}^{(w)})_{|\boldsymbol{\alpha}|=1}^\infty$. Finally, \textit{effect heredity} -- the belief that higher-order effects are active only when all lower-order components are active -- is implicitly imposed via the product structure of \eqref{eq:pod}. For example, the weight for the interaction effect of $x_1$ and $x_2$ depends on the product term $w_1 w_2$, which is large only when both $w_1$ and $w_2$ (corresponding to main effects of $x_1$ and $x_2$) are large as well. A similar framework was also considered by \cite{Jos2006} in the context of experimental design.

We note that the existing literature on POD weights, beginning with \cite{Kea2012} and developed in subsequent papers, focuses largely on solving complex systems of partial differential equations. To contrast, our motivation for the POD form \eqref{eq:pod} is to embed the effect principles from experimental design within the ANOVA-like decomposition of $\mathcal{H}_{\gamma_{\boldsymbol{\theta}}}$. To our knowledge, our work is the first to unify these two fundamental sampling ideas from experimental design and QMC. This unified framework then allows us to derive the sparsity structure on $g$ needed to lift the curse-of-dimensionality for data reduction.

Using the RKHS in \eqref{eq:rkhs} along with Lemma \ref{lem:kh}, the following theorem provides sufficient conditions on $g \in \mathcal{H}_{\gamma_{\boldsymbol{\theta}}}$ required for a dimension-free error rate using PSPs:
\begin{theorem}[Dimension-free error rate]
Assume the ANOVA-like coefficients for $g \in \mathcal{H}_{\gamma_{\boldsymbol{\theta}}}$ follow the POD form \eqref{eq:pod}, with fixed product weights $(w_l)_{l=1}^\infty$ and order weights $(\mathrm{T}_k^{(w)})_{k=1}^\infty$. For fixed $\boldsymbol{\theta} = (\theta_l)_{l=1}^p$, let $F_n$ be the e.d.f. of the $\boldsymbol{\theta}$-weighted PSPs under $\gamma_{\boldsymbol{\theta}}$ in \eqref{eq:aniso}. If:
\begin{equation}
\mathrm{T}_{|\boldsymbol{\alpha}|}^{(w)} = \mathcal{O}\left\{p^{-1/4}\left(|\boldsymbol{\alpha}|!\right)^{-1/2}\right\} \quad \text{and} \quad \sum_{l=1}^\infty w_l^4/\theta_l^2 < 4,
\label{eq:cond}
\end{equation}
then $I(g;F,F_n) \leq {C}/{\sqrt{n}}$ for some constant $C > 0$ not depending on $p$.
\label{thm:dim}
\end{theorem}
\noindent Recall that a dimension-free rate provides a strong form of tractability. Viewed this way, the two conditions in \eqref{eq:cond} shed light on what sparsity structure is needed on the downstream map $g$ for effective high-dim. reduction. Consider first the condition $\mathrm{T}_{|\boldsymbol{\alpha}|}^{(w)} = \mathcal{O}\left\{p^{-1/4} \left( |\boldsymbol{\alpha}| ! \right)^{-1/2}\right\}$ in \eqref{eq:cond}. From the earlier connection between POD weights and the effect principles, this can be broken down as (a) $\mathcal{O}(p^{-1/4})$ -- an \textit{effect sparsity} rate for order importance, controlling the number of active orders in $g$, and (b) $\mathcal{O}\{\left( |\boldsymbol{\alpha}| ! \right)^{-1/2}\}$  -- an \textit{effect hierarchy} rate, dictating the decaying rate of order importance in $g$. (A similar factorial order decay also arises in the dimension-free rate of component-by-component lattice rules, see pg. 76 of \citealp{Dea2013}). Consider next the condition $\sum_{l=1}^\infty w_l^4/\theta_l^2 < 4$ in \eqref{eq:cond}, which can be seen as an \textit{effect sparsity} rate for variable importance. To see this, suppose the simple case of $\theta_l =1$ for all $l$. The resulting constraint $\sum_{l=1}^\infty w_l^4 < 4$ then limits the number of active variables in $g$ (since product weights $(w_l)_{l=1}^\infty$ measure variable importance) -- this is precisely effect sparsity. Moreover, one can counteract an influential variable $l$, i.e., with product weight $w_l \gg 0$, by setting a sufficiently large scale parameter $\theta_l$. This provides a theoretical justification for the earlier observation in Section \ref{sec:prob}, that larger $\theta_l$'s impose greater importance on variable $l$ in reduction.

It is worth emphasizing that, while the rate in $n$ for Theorem \ref{thm:dim} is only the Monte Carlo rate of $\mathcal{O}(n^{-1/2})$\footnote{Technically, this is slightly better than the \textit{almost-sure} Monte Carlo rate of $\mathcal{O}(n^{-1/2}\sqrt{\log \log n})$; see \cite{Kie1961}.}, the importance of this theorem is that it sheds light on the sparsity structure of $g$ (via the effect principles), needed to provide relief from the curse-of-dimensionality for data reduction. As in QMC, this relief is achieved here via the strong tractability requirement of a dimension-free error rate.

\subsection{Error rate for fixed dimension $p$}
\label{sec:conv}

For fixed $p$, the theorem below shows PSPs enjoy an improved rate in $n$ over Monte Carlo:

\begin{theorem}[Fixed-dimension error rate]
Let $\mathcal{A} \subseteq [p]$ be an active set, and let $\boldsymbol{\theta}$ satisfy $\theta_l > 0$ for $l \in \mathcal{A}$ and $\theta_l = 0$ otherwise. Suppose $\mathcal{X} \subseteq \mathbb{R}^p$ is measurable with positive Lebesgue measure, with $F$ satisfying the mild moment condition:
\begin{equation}
\exists \beta > 0, C\geq 0 \text{ s.t. }\limsup_{r \rightarrow \infty} r^\beta \int_{\mathcal{X} \setminus B_r(\bm{y})} \mathbb{E}_{\bm{Y} \sim F}[\gamma_{\boldsymbol{\theta}}(\bm{x},\bm{Y})] \; dF(\bm{x}) \leq C, \;  \text{ for all $\bm{y} \in \mathcal{X}$}.
\label{eq:moment}
\end{equation}
Then, with $\zeta = \beta/(\beta+1)$ and $F_n$ as defined in Theorem \ref{thm:dim}, it follows that for any $\nu \in (0,\zeta)$:
\begin{equation}
\sup_{g \in \mathcal{H}_{\gamma_{\boldsymbol{\theta}}}, \|g\|_{\gamma_{\boldsymbol{\theta}} \leq 1}}I(g;F,F_n) \leq \mathcal{O}\left\{ n^{-1/2}(\log n)^{-(\zeta-\nu)/(2|\mathcal{A}|)} \right\},
\label{eq:fixedrate}
\end{equation}
where constants may depend on $p$ and $\nu$.
\label{thm:pspconv}
\end{theorem}

\noindent We make two remarks here. First, when $g$ is active in all dimensions (i.e., $\mathcal{A}=[p]$) and $F$ is not too heavy-tailed (i.e., it satisfies \eqref{eq:moment}), PSPs enjoy a faster error rate to Monte Carlo by at least the log-factor $(\log n)^{-1/(2p)}$. While this yields a slight theoretical improvement, simulations and applications later on suggest a quicker error rate for PSPs. Indeed, a key gap for kernel sampling methods (see, e.g., \citealp{Cea2012, Bea2012,Bea2015}) is that (a) theory guarantees only a $\mathcal{O}(n^{-1/2})$ rate for infinite-dim. kernels, but (b) empirical performance suggests a $\mathcal{O}(n^{-1})$ rate in practice. Viewed this way, Theorem \ref{thm:pspconv} provides a slight improvement for PSPs over existing rates, and for fixed $p$, the same empirical rate of $\mathcal{O}(n^{-1})$ is observed for PSPs in simulations (see Section \ref{sec:sim}). We therefore use this $\mathcal{O}(n^{-1})$ rate for practical cost comparisons later on.

Second, when PSPs are constructed on $\mathcal{A} \subseteq [p]$ (the active dimensions of $g$), the log-factor in Theorem \ref{thm:pspconv} becomes $(\log n)^{-1/(2|\mathcal{A}|)}$. This improves upon Theorem 5 of \cite{MJ2016b}, in terms of convergence rate in $n$ for high-dim. integration. In practice, however, $\mathcal{A}$ is typically not known a priori. An adaptive scheme can fully exploit the result in Theorem \ref{thm:pspconv}, by iteratively (a) identifying active dimensions, then (b) sequentially targeting these dimensions in reduction. Given the scope of the current paper, we defer this to future work.




\section{SpIn kernel specification}
\label{sec:bayesian}
With this framework in hand, we now investigate the specification for the SpIn kernel $\gamma_{\boldsymbol{\theta} \sim \pi}$ in \eqref{eq:gamma} and \eqref{eq:podtheta}, in terms of its product weights $(\theta_l)_{l=1}^p$ and order weights $(\Gamma_{|\bm{u}|}^{(\theta)})_{|\bm{u}|=1}^\infty$. We first examine a good prior choice $\pi$ for product weights, appealing to an interesting connection to experimental design, then give a brief discussion on order weights.

\subsection{Product weights}
\label{sec:prior}
Consider first the product weights $(\theta_l)_{l=1}^p$ in the POD form \eqref{eq:podtheta}, which, as shown in Section \ref{sec:theory}, quantify variable importance for data reduction. Prior to data, the only knowledge we have on $g$ is that, in high-dimensions, it is likely active for only a small fraction of variables. As seen earlier, one way to quantify this sparsity is via i.i.d. priors on $(\theta_l)_{l=1}^p$. To this end, we consider the following Gamma priors:
\begin{equation}
\theta_l \distas{i.i.d.} \text{Gamma}(\nu, \lambda), \quad \text{i.e.,} \quad \pi(\{\theta_l\}_{l=1}^p) = \prod_{l=1}^p \left\{ \frac{\lambda^{\nu}}{\Gamma(\nu)} \theta_l^{\nu-1} \exp\left(-\lambda \theta_l \right) \right\}.
\label{eq:prior}
\end{equation}
This choice of i.i.d. Gamma priors offers two appealing properties for the SpIn kernel $\gamma_{\boldsymbol{\theta} \sim \pi}$. First, the shape hyperparameter controls the concentration of prior $\pi$ around $\bm{0}$; a smaller choice of $\nu \in (0,1)$ pushes all product weights $(\theta_l)_{l=1}^p$ closer to zero, which results in a more stringent sparsity assumption for the SpIn kernel. Second, this Gamma specification gives a closed-form expression for $\gamma_{\boldsymbol{\theta} \sim \pi} = \mathbb{E}_{\boldsymbol{\theta} \sim \pi} [ \gamma_{\boldsymbol{\theta}} ]$, which is valuable for two reasons: (a) it reveals the role of the scale hyperparameter $\lambda$ for data reduction, and (b) it provides an insightful connection between the proposed PSPs and recent work in experimental design. While the following discussion focuses on the Gamma priors in \eqref{eq:prior}, our algorithm in Section \ref{sec:algo} can be used for any prior $\pi$ which can be efficiently sampled.


Under prior $\pi$ in \eqref{eq:prior}, the (simplified) SpIn kernel has the following closed form:
\begin{proposition}[Closed form SpIn kernel]
Let $\gamma_{\boldsymbol{\theta}}$ be the simplified kernel in \eqref{eq:aniso}. Under prior $\pi$ in \eqref{eq:prior}, its rescaled SpIn kernel becomes:
\begin{align}
\begin{split}
\frac{\gamma_{\boldsymbol{\theta} \sim \pi}(\bm{x},\bm{y})}{\lambda^{\nu p}} = \left( \prod_{l=1}^p \frac{1}{(x_l - y_l)^2 + \lambda} \right)^{\nu}.
\label{eq:closed}
\end{split}
\end{align}
\label{prop:closed}
\vspace{-0.4cm}
\end{proposition}
\noindent The closed-form SpIn kernel in \eqref{eq:closed} can be seen as a product of univariate, inverse multi-quadric kernels \citep{Mic1984}. A closer inspection of \eqref{eq:closed} reveals why this is indeed effective for high-dim. reduction. Let $\bm{x}$ and $\bm{y}$ be two points with (a) the same coordinate $x_l = y_l$ for some variable $l$, and (b) different coordinates for other variables. Plugging these points into \eqref{eq:closed} with $\lambda$ small, the $l$-th inverse-distance term in \eqref{eq:closed} becomes very large, which results in a large value for $\gamma_{\boldsymbol{\theta} \sim \pi}(\bm{x},\bm{y})$. Put another way, the SpIn kernel in \eqref{eq:closed} provides a measure of low-dim. similarity for high-dim. points -- it assigns \textit{high} similarity to two points close in some coordinate, despite other coordinates being vastly different. This property is not enjoyed by radial basis kernels, which may explain the poor performance of herding and SPs for high-dim. reduction (see Figure \ref{fig:projill}).

This closed form also sheds light on the role of hyperparameter $\lambda$ for data reduction. For small $\lambda > 0$, it is clear from \eqref{eq:closed} that the SpIn kernel $\gamma_{\boldsymbol{\theta} \sim \pi}$ magnifies low-dim. similarities between $\bm{x}$ and $\bm{y}$. However, if $\lambda$ is set too close to 0, the resulting kernel becomes nearly singular, which incurs numerical instabilities in optimization. Similar instabilities arise when the hyperparameter $\nu$ (which controls sparsity, see earlier) is set too small. To this end, we found that the hyperparameter setting $(\nu, \lambda) = (0.1,0.01)$ works well for the numerical examples in Section \ref{sec:app}, given that the underlying big data is scaled to zero mean and unit variance for each variable.

Further insight can be gained by plugging the SpIn kernel $\gamma_{\boldsymbol{\theta}}$ in \eqref{eq:closed} into kernel discrepancy \eqref{eq:kerdis}. Setting $\nu =1$, and expanding the terms in \eqref{eq:kerdis}, we get:
\begin{equation}
D^2_{\gamma_{\boldsymbol{\theta} \sim \pi}}(F,F_n) = C + \frac{1}{n^2} \sum_{i=1}^n \sum_{j=1}^n \left( \prod_{l=1}^p \frac{1}{(x_{i,l} - x_{j,l})^2 + \lambda} \right) - \frac{2}{n} \sum_{i=1}^n \mathbb{E}_{\bm{Y} \sim F}\left( \prod_{l=1}^p \frac{1}{(x_{i,l} - Y_l)^2 + \lambda} \right),
\label{eq:mpdisc}
\end{equation}
where $C$ is a constant with respect to $\mathcal{D} := \{\bm{x}_i\}_{i=1}^n$. Recall that the PSPs under $\pi$ minimize discrepancy $D_{\gamma_{\boldsymbol{\theta} \sim \pi}}(F,F_n)$, meaning it jointly minimizes the middle term in \eqref{eq:mpdisc} and maximizes the last term. Setting $\lambda = 0$, the middle term in \eqref{eq:mpdisc} reduces to the maximum projection (MaxPro) criterion in \cite{Jea2015}. By minimizing this criterion, the resulting designs can be shown to be ``space-filling'' on projections, in that no two points are too close to each other on any coordinate projections. In this sense, PSPs can be seen as a novel extension of MaxPro designs for data reduction -- it allows us to \textit{design} a reduced dataset, which captures low-dim. features in high-dim. and non-uniform big data.


\subsection{Order weights}
\label{sec:order}
Consider next the order weights $(\Gamma^{(\theta)}_{|\bm{u}|})_{|\bm{u}|=1}^\infty$ in the POD form \eqref{eq:podtheta}. By effect hierarchy, lower-order interactions are more significant than higher-order ones, so $(\Gamma^{(\theta)}_{|\bm{u}|})_{|\bm{u}|=1}^\infty$ should form a decreasing sequence in $|\bm{u}|$. From the dimension-free rate in Theorem \ref{thm:dim}, we know that a factorial order decay in $g$ offers relief from the curse-of-dimensionality, so the order weights here should decay at least factorially to yield effective high-dim. reduction. We found that an exponential decay $\Gamma^{(\theta)}_{|\bm{u}|} = \exp\{ - |\bm{u}|\}$ gives good performance in practice. Of course, the algorithm in the next section can be used for any choice of order weights.


\subsection{Visualization}
\begin{figure}[t]
\centering
\includegraphics[width=0.8\textwidth]{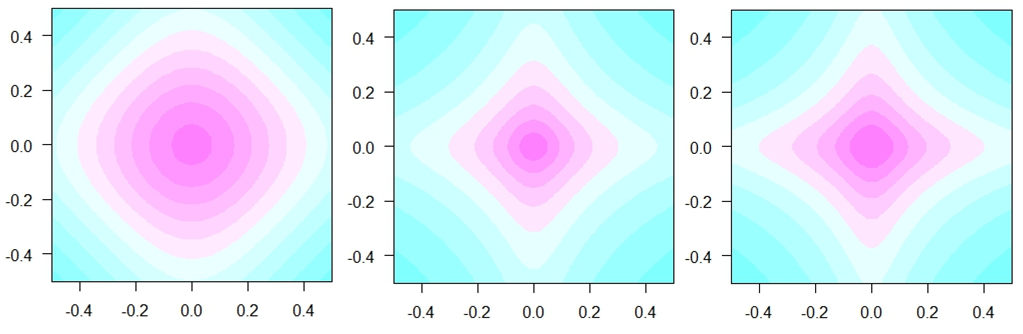}
\caption{Contours of the SpIn kernel with \textup{(left)} $(\nu,\lambda) = (0.1,0.1)$, $(\Gamma^{(\theta)}_1,\Gamma^{(\theta)}_2) = (1,0.5)$, \textup{(middle)} $(\nu,\lambda) = (0.1,0.01)$, $(\Gamma^{(\theta)}_1,\Gamma^{(\theta)}_2) = (1,0.5)$, and \textup{(right)} $(\nu,\lambda) = (0.1,0.01)$, $(\Gamma^{(\theta)}_1,\Gamma^{(\theta)}_2) = (1,0.01)$.}
\label{fig:kp}
\end{figure}
We provide a simple visualization to motivate the above SpIn kernel specification. Figure \ref{fig:kp} (left, middle, right) shows the contours of the SpIn kernel with $(\nu,\lambda) = (0.1,0.1)$ and $(\Gamma^{(\theta)}_1,\Gamma^{(\theta)}_2) = (1,0.5)$, $(\nu,\lambda) = (0.1,0.01)$ and $(\Gamma^{(\theta)}_1,\Gamma^{(\theta)}_2) = (1,0.5)$, and $(\nu,\lambda) = (0.1,0.01)$ and $(\Gamma^{(\theta)}_1,\Gamma^{(\theta)}_2) = (1,0.01)$, respectively. First, by changing the hyperparameter $\lambda$ from 0.1 to 0.01 (left to middle), we see that the SpIn kernel places greater emphasis on lower-dim. features. Next, by changing the second-order weight $\Gamma^{(\theta)}_2$ from 0.1 to 0.01, the resulting SpIn kernel becomes more aggressive in pursuing coordinate-wise similarities. Both observations are consistent with earlier insights.


\section{Optimization algorithms}
\label{sec:algo}

We now present two algorithms for optimizing the PSPs in \eqref{eq:pipsp} for data reduction. We begin by outlining key components of these algorithms, then provide a practical discussion on running time. For brevity and clarity, we have moved technical details and derivations to the Appendix, and instead focus on important ideas.


\subsection{Algorithm sketch}


\subsubsection{\texttt{psp.mm} -- One-shot reduction}

We first sketch out the key steps behind \texttt{psp.mm}, a one-shot reduction algorithm for optimizing the ($\pi$-expected) PSPs in \eqref{eq:pipsp}. Given a specification of choice for (a) order weights $(\Gamma_{|\bm{u}|}^{(\theta)})_{|\bm{u}|=1}^\infty$ and (b) prior $\pi$ for product weights $(\theta_l)_{l=1}^p$, the desired optimization problem (from \eqref{eq:pipsp} and \eqref{eq:kerdis}) can be restated as:
\begin{equation}
\argmin_{\mathcal{D} = \{\bm{x}_1, \cdots, \bm{x}_n\}} \left[  \frac{1}{n^2} \sum_{i=1}^n \sum_{j=1}^n \mathbb{E}_{\boldsymbol{\theta} \sim \pi}\left\{ \gamma_{\boldsymbol{\theta}} (\bm{x}_i,\bm{x}_j) \right\} - \frac{2}{n} \sum_{i=1}^n \mathbb{E}_{\bm{Y} \sim F_N,\boldsymbol{\theta} \sim \pi}\left\{ \gamma_{\boldsymbol{\theta}}(\bm{x}_i,\bm{Y}) \right\} \right],
\label{eq:pspopt}
\end{equation}
\noindent where $\mathcal{D}= \{\bm{x}_i\}_{i=1}^n$ is the reduced point set. Note that, in \eqref{eq:pspopt}, the data generating distribution $F$ (which was used for theoretical analysis in previous sections) is replaced by the big data realization $F_N$ (which is available in practice).

There are two challenges for optimizing \eqref{eq:pspopt}. First, evaluating the objective is computationally expensive, because (a) the expectation over $\bm{Y} \sim F_N$ involves a massive summation over every big data point, and (b) for a \textit{general} POD specification, there is no closed form for the SpIn kernel $\mathbb{E}_{\boldsymbol{\theta} \sim \pi}\left\{ \gamma_{\boldsymbol{\theta}}\right\}$. Second, the optimization in \eqref{eq:pspopt} is not only high-dim., but also non-convex as well. To tackle this, we will first optimize each point $\bm{x}_i$ (given remaining points) using the following two-step procedure, then cycle this procedure over all $n$ points until the point set convergences.

\begin{figure}[!t]
\begin{minipage}{0.49\textwidth}
\begin{algorithm}[H]
\small
\caption{\texttt{psp.mm}: One-shot PSPs}
\label{alg:pspsccp}
\begin{algorithmic}
\stb Warm-start the initial point set $\mathcal{D}^{[0]} \leftarrow \{\bm{x}_i^{[0]}\}_{i=1}^n$ using SPs. Set $l=0$.
\stb \textbf{Repeat} until convergence of $\mathcal{D}^{[l]}$:
\bi[leftmargin=30pt]
\item \textbf{For} $i = 1, \cdots, n$:
\bi[leftmargin=10pt]
\item Resample $\mathcal{Y} \distas{i.i.d.} F_N$ and $\vartheta \distas{i.i.d.} \pi$.
\item Set $\bm{x}_i^{[l+1]} \leftarrow \mathcal{M}_i(\bm{x}_i^{[l]}; \mathcal{Y}, \vartheta, \mathcal{D}^{[l]}_{-i})$, with $\mathcal{M}_i$ defined in \eqref{eq:closedform2}.
\item Update $\mathcal{D}^{[l]}_{i} \leftarrow \bm{x}_i^{[l+1]}$.
\ei
\item Update $\mathcal{D}^{[l+1]} \leftarrow \{\bm{x}_i^{[l+1]}\}_{i=1}^n$, and set $l \leftarrow l + 1$.
\ei
\stb Return the point set $\mathcal{D}^{[\infty]}$.
\end{algorithmic}
\end{algorithm}
\end{minipage}
\hfill
\begin{minipage}{0.49\textwidth}
\begin{algorithm}[H]
\small
\caption{\texttt{psp.mm.seq}: Sequential PSPs}
\label{alg:pspsccpseq}
\begin{algorithmic}
\stb Initialize first point $\bm{x}_1 \distas{i.i.d.} F_N$, $\mathcal{D} \leftarrow \{\bm{x}_1\}$.
\stb For $i = 2, \cdots, n$:
\bi[leftmargin=30pt]
\item Set $l=0$ and $\bm{x}_i^{[0]} \distas{i.i.d.} F_N$.
\item \textbf{Repeat} until convergence of $\bm{x}_i^{[l]}$:
\bi[leftmargin=30pt]
\item Resample $\mathcal{Y} \distas{i.i.d.} F_N$ and $\vartheta \distas{i.i.d.} \pi$.
\item Set $\bm{x}_i^{[l+1]} \leftarrow \mathcal{M}_i(\bm{x}_i^{[l]}; \mathcal{Y}, \vartheta, \mathcal{D}_{-i})$, with $\mathcal{M}_i$ defined in \eqref{eq:closedform2}.
\item Set $l \leftarrow l+1$.
\ei
\item Update $\mathcal{D} \leftarrow \mathcal{D} \cup \{\bm{x}_i^{[\infty]}\}$.
\ei
\stb Return the point set $\mathcal{D}$.
\end{algorithmic}
\end{algorithm}
\end{minipage}
\end{figure}

Consider first the optimization of point $\bm{x}_i$ given remaining points $\mathcal{D}_{-i}$. The first step in \texttt{psp.mm} is to take an \textit{unbiased} estimate of the objective in \eqref{eq:pspopt}, using small subsamples from both big data $F_N$ and prior $\pi$. Denoting these subsamples as $\mathcal{Y} = \{\bm{y}_m'\}_{m=1}^{N_s} \distas{i.i.d.} F_N$, $N_s \ll N$ and $\vartheta = \{\boldsymbol{\theta}_r\}_{r=1}^R \distas{i.i.d.} \pi$, this estimate becomes:
\begin{equation}
\frac{1}{n^2R} \sum_{i=1}^n \sum_{j=1}^n \sum_{r=1}^R \gamma_{\boldsymbol{\theta}_r} (\bm{x}_i,\bm{x}_j) - \frac{2}{nN_sR} \sum_{i=1}^n \sum_{m=1}^{N_s} \sum_{r=1}^R \gamma_{\boldsymbol{\theta}_r}(\bm{x}_i,\bm{y}_m').
\label{eq:pspoptest}
\end{equation}
In contrast to \eqref{eq:pspopt}, the estimated objective \eqref{eq:pspoptest} can now be efficiently evaluated and optimized. This subsampling step is motivated by the success of stochastic algorithms (e.g., \citealp{Bot2010}) for speeding up large-scale machine learning optimization problems.

The second step in \texttt{psp.mm} is to minimize the unbiased objective estimate \eqref{eq:pspoptest} for point $\bm{x}_i$, given fixed $\mathcal{D}_{-i}$. To do this, we make use of a non-linear optimization method called \textit{majorization-minimization} (MM) \citep{Lan2016}, which is widely used in statistical learning. The intuition is as follows: to minimize a non-linear function $f$, MM iteratively minimizes a surrogate function $h$ which lies above $f$ ($h$ is known as a \textit{majorizer}). One appealing feature of MM is the so-called \textit{descent property} \citep{Lan2016}, which ensures solution iterates are always decreasing for $f$, the desired function to minimize. The key to computational efficiency is to choose the majorizer $h$ so that it admits an easy-to-evaluate solution, which then speeds up MM iterations. For \eqref{eq:pspoptest}, a nice quadratic majorizer can be derived by exploiting the Gaussian kernel, yielding a closed-form map $\mathcal{M}_i(\cdot; \mathcal{Y},\vartheta, \mathcal{D}_{-i})$ for minimizing \eqref{eq:pspoptest} (the exact form of $\mathcal{M}_i$ is tedious, and is provided in \eqref{eq:closedform2} of the Appendix). In our experience, MM works much better than gradient descent methods here. One reason is that the former guarantees descent to a good solution in practical time, whereas the latter requires many objective and gradient evaluations, which is time-consuming with big data.

Algorithm \ref{alg:pspsccp} summarizes the above two-step procedure for \texttt{psp.mm}. For each point $\bm{x}_i$, we first subsample the big data $F_N$ and prior $\pi$, then update $\bm{x}_i$ by applying the closed-form optimization map $\mathcal{M}_i$ once. These two steps are then cycled over all $n$ points (a technique known as \textit{blockwise coordinate descent}; \citealp{Tse2001}) until the point set $\mathcal{D}$ converges. This blockwise descent allows us to exploit the pointwise optimization structure in \eqref{eq:pspopt} for efficient reduction. The following theorem gives a convergence guarantee for \texttt{psp.mm}:
\begin{theorem}[Convergence of \texttt{psp.mm}]
Suppose $\mathcal{X}$ is convex and compact. For any initial point set $\mathcal{D}^{[0]} \subseteq \mathcal{X}$, the sequence $(\mathcal{D}^{[l]})_{l=1}^\infty$ from \texttt{\textup{psp.mm}} converges almost surely to a stationary limiting point set $\mathcal{D}^{[\infty]}$ for \eqref{eq:pspopt}.
\label{thm:algconv}
\end{theorem}
\noindent In other words, when the sample space is convex and compact, \texttt{psp.mm} converges to a stationary solution for the desired problem \eqref{eq:pspopt}. For \textit{finite} big data, this compactness condition is trivially satisfied. For the broader problem of compacting a \textit{distribution} $F$, this algorithm appears to work well even when the compactness condition is violated for $\mathcal{X}$.

\subsubsection{\texttt{psp.mm.seq} -- Sequential reduction}
In practice, the one-shot algorithm \texttt{psp.mm} works well for a small sample size $n$. For larger $n$, however, a greedy, sequential optimization of \eqref{eq:pspopt} can be more computationally efficient (at the cost of losing the optimality guarantee in Theorem \ref{thm:algconv}). This sequential reduction has two additional advantages: \textit{monotonicity} -- a more monotone error decay as $n$ increases, and \textit{extensibility} -- the ability to reuse prior data if more points are added later on. The latter is particularly important when downstream computations are expensive, since it avoids having to re-evaluate $g$ when additional points are needed. To this end, Algorithm \ref{alg:pspsccpseq} presents a greedy PSP reduction algorithm, called  \texttt{psp.mm.seq}, which is essentially a sequential implementation of the earlier two-step resampling-descent procedure. \texttt{psp.mm.seq} can also be viewed as a specific instance of the general herding scheme in \eqref{eq:herd}:
\begin{equation}
\bm{x}_{n+1} \leftarrow \underset{\bm{x} \in \mathcal{X}}{\textup{Argmax}} \left\{ \mathbb{E}_{\bm{Y} \sim F}[\gamma_{\boldsymbol{\theta}\sim \pi}(\bm{x},\bm{Y})] - \frac{1}{n+1} \sum_{i=1}^n \gamma_{\boldsymbol{\theta}\sim \pi}(\bm{x},\bm{x}_i) \right\},
\label{eq:herdpsp}
\end{equation}
with the novelty here being the use of the new SpIn kernel $\gamma = \gamma_{\boldsymbol{\theta}\sim \pi}$ for capturing low-dim. features. From a running time perspective (see below), we recommend the sequential method \texttt{psp.mm.seq} over the one-shot method \texttt{psp.mm} when $n \geq 1,000$ or $p \geq 100$.

\subsection{Algorithm running time}
Next, we investigate the running times for the one-shot and sequential PSP algorithms (\texttt{psp.mm} and \texttt{psp.mm.seq}), and provide insight on when each method should be used in practice. This requires the analysis of two steps: (a) the computation of scale parameters $\boldsymbol{\theta} = (\theta_{\bm{u}})_{|\bm{u}|=1}^p$ from product weights $(\theta_l)_{l=1}^p$ and order weights $(\Gamma_{|\bm{u}|}^{(\theta)})_{|\bm{u}|=1}^\infty$, and (b) the evaluation of the MM map $\mathcal{M}_i$. For (a), a brute-force evaluation of $\boldsymbol{\theta}$ requires $\mathcal{O}(2^p)$ work, which is clearly infeasible for moderate $p$. Motivated by the recursive construction of POD-weighted shifted lattice rules (Section 5.6 in \citealp{Dea2013}), we implement here a similar recursive procedure, which computes necessary information on $\boldsymbol{\theta}$ for PSP optimization using $\mathcal{O}(p^2)$ work; details on this in Appendix \ref{app:rec}. For (b), the evaluation of $\mathcal{M}_i$ (see \eqref{eq:closedform2} in Appendix) requires $\mathcal{O}(np)$ work, assuming subsample sizes $N$ and $R$ are independent of $n$ and $p$. From this, it follows that (a) the running time for one loop iteration\footnote{To achieve some optimality gap $\epsilon$, the number of loop iterations may also depend on $n$ and $p$, but this dependence is difficult to establish for nonlinear optimization problems \citep{NW2006}. In our numerical examples, a fixed number of iterations (say, 200), along with convergence checks on $\mathcal{D}$, works quite well; we therefore use this per-iteration cost for analyzing running time.} of \texttt{psp.mm} is $\mathcal{O}\{n(np+p^2)\}$, and (b) the running time of \texttt{psp.mm.seq} for a new $i$-th point (again, for one loop iteration) is $\mathcal{O}(ip+p^2)$. In practice, the sequential algorithm \texttt{psp.mm.seq} is much faster to perform than the one-shot algorithm \texttt{psp.mm}, and allows for more efficient reduction in problems with large $n$ or $p$.

These running times also offer a more complete view on the advantages and limitations of PSPs, compared to random sampling of big data. We analyze this via the following marginal cost trade-off, i.e., by comparing (a) the marginal cost required for computing an additional PSP point, with (b) its corresponding marginal cost savings for downstream computations. For (a), we know (from the discussion above) that the marginal cost for computing an $n$-th PSP point is $\mathcal{O}(n)$ and $\mathcal{O}(n^2)$, for the sequential and one-shot algorithms. For (b), first let $C_n$ be the marginal cost increase for downstream computations with an additional $n$-th point (see the kernel learning application in Section \ref{sec:app}, for an example where $C_n$ is not constant). For fixed dimension $p$, the practical error gain of PSPs over Monte Carlo is $\mathcal{O}(n^{1/2})$ (see Section \ref{sec:conv}), so the marginal cost savings of an additional $n$-th PSP point is $\mathcal{O}(C_n n^{1/2})$ in practice. Comparing the marginal costs in (a) and (b), it follows that for expensive downstream computations (i.e., when $C_n \geq K n^{1/2}$ (or $C_n \geq K n^{3/2}$) for the sequential (or one-shot) approach, where $K$ is a constant), PSPs can yield improved reduction over random sampling in terms of the marginal cost trade-off.

We emphasize here that, as described in Section \ref{sec:prob}, an implicit assumption for data reduction is that downstream computations are indeed expensive, and \textit{cannot} be performed for the full data within a practical budget (e.g., on time or resources). This is an often-encountered scenario in real-world statistical problems: (a) in engineering statistics, the propagation of input parameters through complex simulation models can require weeks or even months of computation \citep{Mea2017,Yea2018}; (b) in machine learning, model fitting can incur high computation and memory costs for large datasets \citep{Fea2001}. As we show next, it is within this context where PSPs can yield improved reduction over both random sampling and existing methods, particularly for high-dim. data.

\section{Numerical examples}

\label{sec:app}

With this in hand, we now investigate the performance of PSPs in numerical examples. We begin by presenting some simulation studies, then explore the effectiveness of PSPs in two real-world data reduction applications for kernel learning and MCMC reduction.

%

\subsection{Simulations}
\label{sec:sim}

We first motivate the numerical effectiveness of PSPs via a toy example. Setting $F$ as the 2-d i.i.d. $\text{Beta}(2,4)$ distribution, Figure \ref{fig:pairs} plots the $n=25$-point SPs and PSPs for reducing the data-generating measure $F$. While SPs provide a (visually) good summary of $F$ in the full 2-d space, it yields a lackluster representation of its two marginal distributions. The proposed PSPs, on the other hand, offer a good representation of \textit{both} the full 2-d distribution as well as its marginals, which shows the effectiveness of the sparsity-inducing kernel $\gamma_{\boldsymbol{\theta} \sim \pi}$ in incorporating sparsity for data reduction. In problems where $g$ has only 1 of 2 variables (in general, a small fraction of $p$ variables) active, Figure \ref{fig:pairs} shows how the proposed PSPs can offer improved reduction over existing methods.

\begin{figure}[t]
\centering
\includegraphics[width=0.65\textwidth]{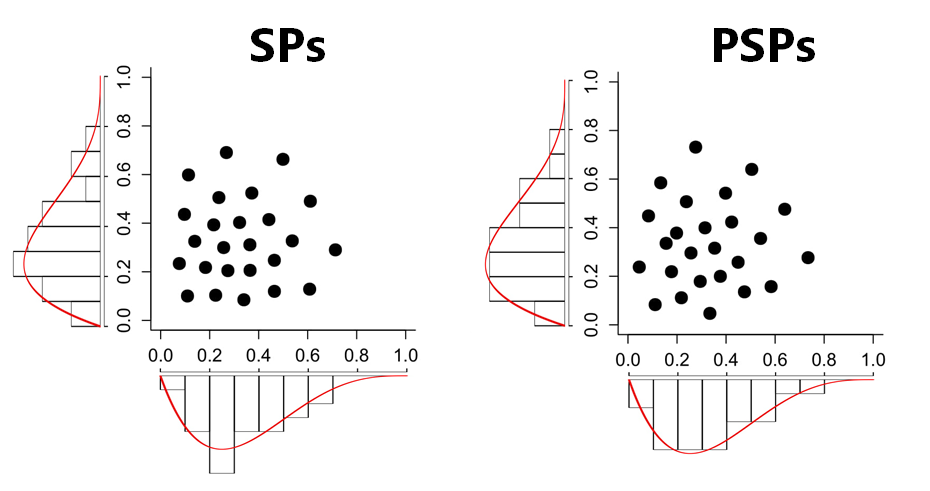}
\caption{Scatterplots and marginal histograms of $n=25$-point SPs and PSPs for $F = $ 2-d i.i.d. $Beta(2,4)$ distribution. True marginal densities are marked in \textbf{\crd{red}}.}
\label{fig:pairs}
\end{figure}



Next, we compare the performance of PSPs with random sampling (Monte Carlo), and herding (using $\gamma$ as the std. Gaussian kernel), SPs, and a QMC sampling method called inverse-Sobol' points. The latter is generated by first sampling from scrambled (or randomized) Sobol' sequence \citep{Sob1967, Owe1998}, then mapping these points via the inverse-transform of $F$. Note that these inverse-Sobol' points are \textit{not} obtainable in real-world reduction problems, where one has only finite realizations of big data, not the underlying measure $F$ itself. In this sense, these inverse-Sobol' points provide a good benchmark for both one-shot and sequential PSPs, showing how a well-established QMC method performs if the underlying $F$ were indeed known. Here, we tested three choices of $F$: the i.i.d. $\mathcal{N}(0,1)$, the i.i.d. $\text{Exp}(1)$ and the i.i.d. $\text{Beta}(2,4)$ distributions, with dimension $p$ from $5$ to $100$. For the downstream map $g$, we employed two well-known test functions: the Gaussian peak function (GAPK, \citealp{Gen1984}) $g(\bm{x}) = \exp\left\{- \sum_{l=1}^p \alpha_l^2 (x_l - u_l)^2\right\}$ and the additive Gaussian function (ADD) $g(\bm{x}) =  \exp\left\{- \sum_{l=1}^p \beta_l x_l\right\}$, where $u_l$ is the marginal mean of $F_l$. To incorporate low-dim. structure, a fraction $q$ of the $p$ variables are set as active, with $\alpha_l = \beta_l=0.25/(qp)$ for active variables, and 0 otherwise. These functions are denoted as GAPK[$q$] and ADD[$q$], respectively.

\begin{figure}[t]
\centering
\includegraphics[width=0.9\textwidth]{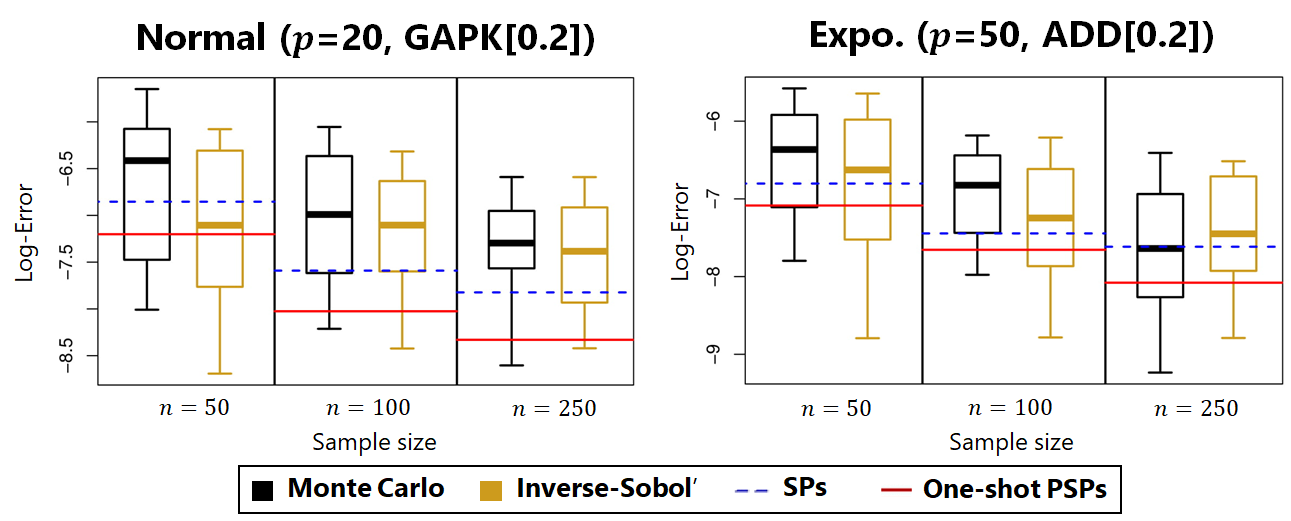}
\caption{Log-errors for \textup{GAPK[0.2]} and \textup{ADD[0.2]} under $F = $ 20-d i.i.d. $\mathcal{N}(0,1)$ and 50-d i.i.d. $\text{Exp}(1)$ distribution. Lines denote log-avg. errors; shaded bands mark 25-th / 75-th quantiles.}
\label{fig:integos}
\end{figure}
\begin{figure}[t]
\centering
\includegraphics[width=0.9\textwidth]{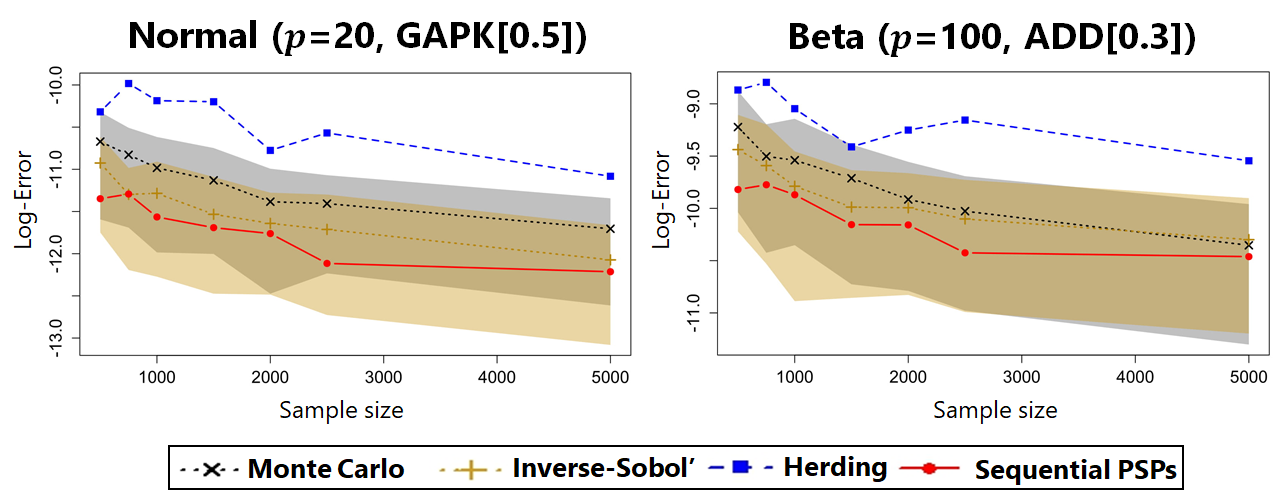}
\caption{Log-errors for \textup{GAPK[0.5]} and \textup{ADD[0.3]} under $F = $ 20-d i.i.d. $\mathcal{N}(0,1)$ and 100-d i.i.d. $\text{Beta}(2,4)$ distribution. Lines denote log-avg. errors; shaded bands mark 25-th / 75-th quantiles.}
\label{fig:integseq}
\end{figure}

We compare first the performance of one-shot methods, with Figure \ref{fig:integos} showing the log-errors $\log I(g;F,F_n)$ in \eqref{eq:interr} for $n=50$, $100$ and $250$ reduced points. In the left figure ($F = 20$-d i.i.d. $\mathcal{N}(0,1)$), the one-shot PSPs (generated with \texttt{psp.mm}) yield noticeably lower errors for GAPK[0.2], compared to random sampling, SPs, and even inverse-Sobol' points (which has access to the true $F$). This shows that the proposed method, in using the new SpIn kernel, allows for effective learning of low-dim. downstream quantities, using only big data sampled from $F$. In the right figure ($F = 50$-d i.i.d. $\text{Exp}(1)$), the one-shot PSPs again give the lowest errors for ADD[0.2]. Similar results hold for other $F$ and $p$, and are omitted for brevity.

We compare next the performance of sequential methods, with Figure \ref{fig:integseq} showing the log-errors for $n = 500 - 5,000$ points. In the left figure ($F = 20$-d i.i.d. $\mathcal{N}(0,1)$), the sequential PSPs (generated with \texttt{psp.mm.seq}) yield considerably lower errors for GAPK[0.5], compared to random sampling, herding and inverse-Sobol' (which has access to $F$). In the right figure ($F = 100$-d i.i.d. $\text{Beta}(2,4)$), the sequential PSPs again give the lowest errors for ADD[0.3]. The comparable performance of PSPs to inverse-Sobol' points supports the claim in Section \ref{sec:conv}, that PSPs indeed enjoy a $\mathcal{O}(n^{-1})$ error rate in practice. Interestingly, herding performs worse than even Monte Carlo, which is not too surprising, since the std. Gaussian kernel measures similarities based on distances in the full $p$-dim. space. To contrast, by using the new SpIn kernel (which accounts for similarities in low-dim. projections, see Figure \ref{fig:projill}), the sequential PSPs yield improved performance to existing methods. This nicely demonstrates the \textit{curse-of-dimensionality} effect for high-dim. reduction if the kernel $\gamma$ is not carefully chosen, and how the proposed SpIn kernel offers practical relief from this curse.

\subsection{Application: Data reduction for kernel ridge regression}
\label{sec:ridge}

\begin{figure}[t]
\centering
\includegraphics[width=0.9\textwidth]{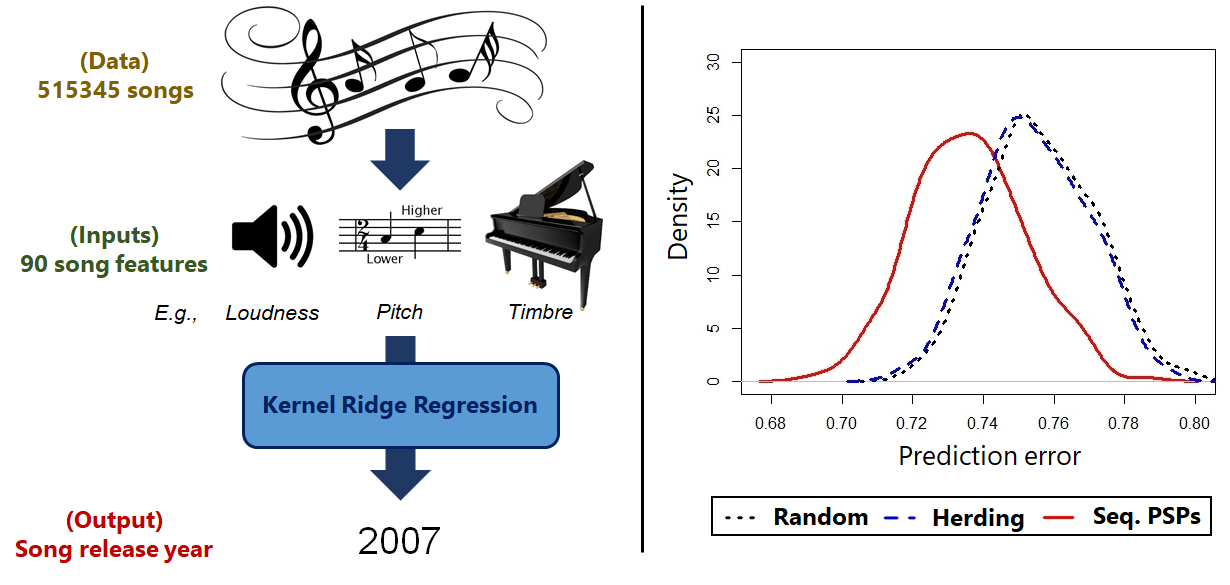}
\caption{\textup{(Left)} A visualization of KRR for predicting song release year. \textup{(Right)} Prediction error densities for 250 randomly chosen songs in testing data.}
\label{fig:msy}
\end{figure}

We now explore the effectiveness of PSPs for data reduction in kernel learning. Kernel learning methods make use of a kernel $k$ to provide effective, non-linear learning for both supervised and unsupervised problems (\citealp{Fea2001,Vap2013}). However, with $N$ denoting the training dataset size, one key bottleneck for kernel methods is that it involves the inverse of an $N \times N$ matrix, which requires $\mathcal{O}(N^3)$ work and $\mathcal{O}(N^2)$ storage. For $N$ large, this computation becomes very time- and memory-consuming (e.g., with $N > 5,000$, this becomes infeasible on many desktop computers). This problem is further compounded when training data is high-dimensional (i.e., $p \gg 1$), since a larger sample size $N$ is required. Here, PSPs can reduce the large, high-dim. training dataset to a smaller dataset which retains low-dim. features for modeling.\footnote{Of course, there have been many methods proposed for big-data kernel learning, most involving some form of Nystr{\"o}m approximation \citep{WS2001}. It is not our aim to compare with all methods in the literature, but to highlight the effectiveness of PSPs as a \textit{data reduction} tool for kernel learning.}

We illustrate this using a well-known machine learning dataset, the Million Song Dataset (MSD; \citealp{Bea2011}). MSD is a collection of audio features, extracted from a million music tracks released in the years 1922 -- 2011. We consider a subset of this data (515,345 songs) from the UCI Machine Learning Repository, with $N = 463,715$ songs for training and the remainder for testing (this split is recommended by data publishers).  In total, $p=90$ song features (continuous) are extracted, including the loudness, pitch, and timbre of each track. The goal is to fit a predictive model using training data, then use this to predict the release year (treated as continuous) for a new song in the testing data. 

To build this predictive model, we employ a kernel method called \textit{kernel ridge regression} (KRR; \citealp{Fea2001}). Given (a) a kernel of choice $k$, and (b) training song features $\{\bm{f}_m\}_{m=1}^N$ (inputs, normalized to zero mean and unit variance) and release years $\{y_m\}_{m=1}^N$ (output, normalized), KRR fits the following non-linear smoother $\hat{h}$:
\begin{equation}
\hat{h} \leftarrow \underset{h \in \mathcal{H}_k}{\textup{Argmin}} \left[ \frac{1}{N} \sum_{m=1}^N \left\{ y_m - h(\bm{f}_m) \right\}^2 + \lambda \|h\|_{k}^2 \right],
\label{eq:krr}
\end{equation}
where $\mathcal{H}_k$ is the RKHS of $k$. The smoother $\hat{h}$ can be viewed as the function in $\mathcal{H}_k$ which best fits the training data, subject to a regularization penalty $\lambda \|h\|_{k}^2$. With this fit, one can then use $\hat{h}(\bm{f}_{new})$ to predict the release year of a new song with features $\bm{f}_{new}$. As typical in statistical problems, the penalty $\lambda$ is tuned via cross-validation \citep{Fea2001}.

The fit in \eqref{eq:krr}, however, requires the inverse of the matrix ${[k(\bm{f}_m,\bm{f}_{m'})]_{m=1}^N}_{m'=1}^N$ \citep{Fea2001}, which incurs $\mathcal{O}(N^3)$ work. For large $N$ (e.g., $N = 463,715$ in MSD), this becomes computationally infeasible! To this end, let $n \ll N$, and consider the \textit{reduced} fit:
\begin{equation}
\hat{h}' \leftarrow \underset{h \in \mathcal{H}_k}{\textup{Argmin}} \left[ \frac{1}{n} \sum_{i=1}^n \left\{ y_i' - h(\bm{f}_i') \right\}^2 + \lambda \|h\|_{k}^2 \right],
\label{eq:krrred}
\end{equation}
where $\mathcal{T}_n' := \{(\bm{f}_i',y_i')_{i=1}^n\}$ is a reduced subset of the full data $\mathcal{T}_N := \{(\bm{f}_m,y_m)_{m=1}^N\}$. Using \eqref{eq:krrred}, the computation time of $\hat{h}'$ reduces from $\mathcal{O}(N^3)$ to $\mathcal{O}(n^3)$. The goal then is to choose a good reduction $\mathcal{T}_n'$, so that the objective in \eqref{eq:krrred} well-approximates that in \eqref{eq:krr}. With $F_N$ denoting the e.d.f. of $\mathcal{T}_N$, this is akin to finding a reduced dataset $\mathcal{T}_n' \subseteq \mathcal{T}_N$ (with e.d.f. $F_n$) such that $\mathbb{E}_{\bm{X} \sim F_n} [g(\bm{X})] \approx \mathbb{E}_{\bm{X} \sim F_N} [g(\bm{X})]$, where $g(\bm{f},y) = \{y - h(\bm{f})\}^2$. It is also highly unlikely that \textit{all} $p=90$ song features are useful for prediction (e.g., from intuition, song pitch and its interaction effects should not be important predictors for release year), which suggests that the desired function $h$ (and hence $g$) is low-dimensional. Viewed this way, the PSPs of $F_N$ should provide a good reduced dataset to use in \eqref{eq:krrred}.

We compare three reduction methods: (a) sequential PSPs on $F_N$ (rounded to closest point in $\mathcal{T}_N$), (b) herding points on $F_N$ (rounded to closest point in $\mathcal{T}_N$), and (c) random subsampling on $\mathcal{T}_N$. All three employ a reduced sample size of $n=4,000$ points (this is set from computation constraints on a standard desktop computer), and are judged on out-of-sample prediction errors for 250 random songs in the testing set; this randomization is then repeated 250 times to measure error variability.

Figure \ref{fig:msy} plots the prediction error densities for the three methods. Two observations are of interest. First, herding provides very little error reduction over random sampling, which is again not too surprising, since the std. Gaussian kernel does not account for low-dim. similarities between points. Second, PSPs offer noticeably better predictive performance over both random sampling and herding, which demonstrates the effectiveness of the SpIn kernel in capturing low-dim. features for predictive modeling.

Lastly, we compare the running times of these methods (both reduction and the KRR computation in \eqref{eq:krrred}, with $\lambda$ tuned via cross-validation), with the hypothetical running time of the full KRR in \eqref{eq:krr} without data reduction. Not surprisingly, random sampling is the quickest method, requiring 1,583 seconds of computation on a single-core 3.4 Ghz processor, while kernel herding and PSPs require 3,423 and 3,965 seconds, respectively. To contrast, the full KRR in \eqref{eq:krr} (with no reduction) has a hypothetical running time of $\mathcal{O}(N^3) = 1,583 \cdot N^3 / n^3\text{ seconds} \approx 78$ years, and requires $\mathcal{O}(N^2) \approx 1,720$ gigabytes of memory, which is clearly infeasible to tackle directly. Given this practical cost constraint, our reduction method offers the best predictive performance of the three methods tested.

\subsection{Application: Reduction of MCMC chains}
\label{sec:mcmc}
Next, we apply PSPs to the important problem of reducing MCMC chains in Bayesian computation. For Bayesian modeling, parameters are learned by sampling from a posterior distribution, with this sampling typically performed via MCMC methods \citep{Gea1995}. In practice, Bayesian practitioners perform a post-processing step called \textit{thinning}, which discards all-but-every $k$-th sample from the MCMC sample chain $\{\Theta_m\}_{m=1}^N$. Thinning is done for three reasons \citep{LE2012}: it reduces high sample autocorrelations, lowers storage requirements, and reduces computation time for downstream computations. One key weakness of thinning is that it is quite wasteful, since valuable information from posterior samples are thrown away. Here, PSPs can offer an improved alternative to thinning, by using the \text{full} MCMC chain to train a good representative point set. Our approach can be particularly effective for Bayesian modeling of engineering problems, where downstream posterior computations often involve expensive, time-consuming experiments.

\begin{figure}[t]
\centering
\includegraphics[width=0.90\textwidth]{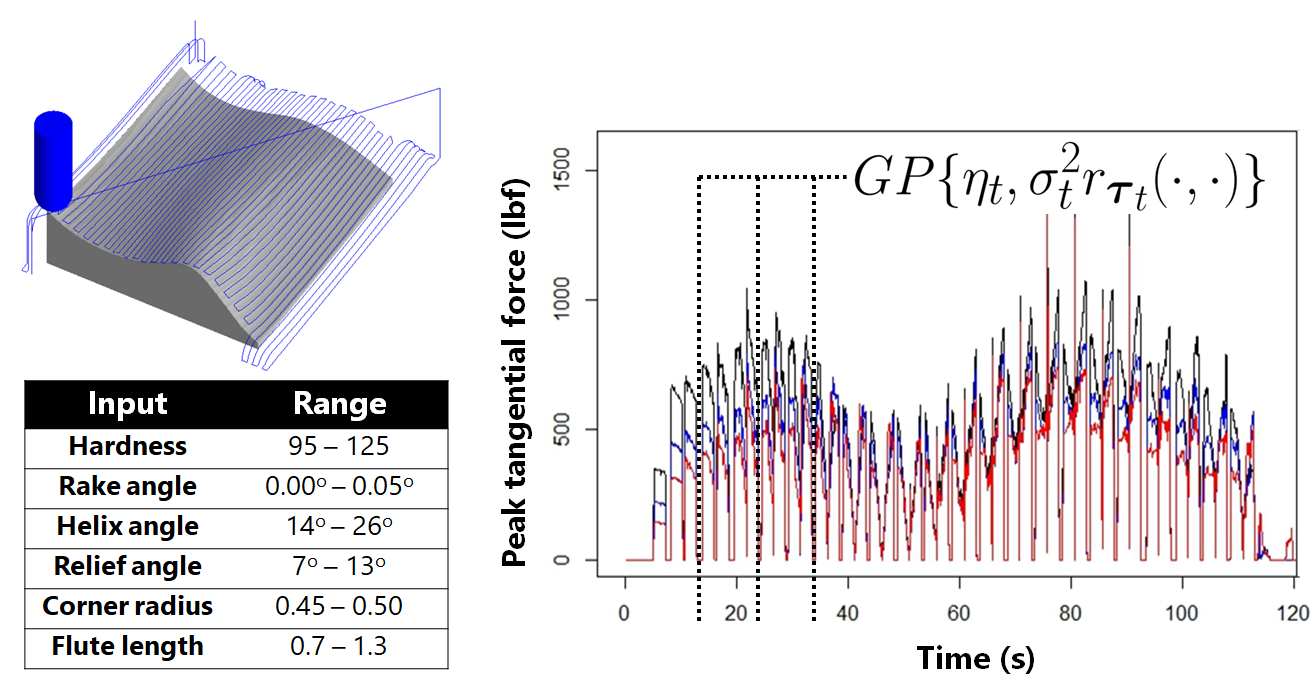}
\caption{\textup{(Top left)} A visualization of the solid end milling procedure. \textup{(Bottom left)} Design inputs and design ranges. \textup{(Right)} Peak tangential force over time for different input settings.}
\label{fig:mill}
\end{figure}


We illustrate this via an application in solid end milling, a cutting process used for precise part machining in the aerospace industry. This is visualized in Figure \ref{fig:mill} (top left): a cutting tool (in blue) is applied at a force to the workpiece (in gray), then moved along the surface on the blue lines, stripping away material as it passes. Figure \ref{fig:mill} (bottom left) gives the desired design region for the six process inputs -- five for the cutting tool, and one for material hardness. Figure \ref{fig:mill} (right) shows, for different input settings, the peak tangential forces over time ($T=3,373$ times in total) -- a key output of the milling process. These forces are simulated via complex computer models on the Production Module software\footnote{https://www.thirdwavesys.com/production-module/}, which is time-consuming to run. Because of the costly nature of simulation experiments, the strategy is to first run a small number of experiments, then use this data to build an \textit{emulator} which efficiently predicts forces at an untested input.


We use the following Gaussian process (GP) emulator model (see \citealp{Sea2013}). For fixed inputs $\bm{c} \in \mathbb{R}^6$, let $f_t(\bm{c})$ be the tangential force at time $t$. Our model assumes:
\begin{equation}
f_t(\bm{c}) \sim \text{GP}\{\eta_t, \sigma^2_t r(\cdot,\cdot; \boldsymbol{\tau}_t)\}, \quad f_t(\bm{c}) \perp f_{t'}(\bm{c}), \quad t \neq t'.
\label{eq:emu}
\end{equation}
In other words, for each time slice, the forces over input space follow independent GPs, with time-dependent mean $\eta_t$, variance $\sigma^2_t$, and length-scale parameters $\boldsymbol{\tau}_t \in \mathbb{R}^6_+$ for the Gaussian correlation $r$. Figure \ref{fig:mill} (right) visualizes this emulation model. Suppose computer experiments are conducted at inputs $\{\bm{c}_d\}_{d=1}^D$, yielding simulated forces $\{f_1(\bm{c}_d), \cdots, f_T(\bm{c}_d)\}_{d=1}^D$. For \textit{fixed} parameters $\Theta_t = \{\eta_t,\sigma^2_t,\boldsymbol{\tau}_t\}$, the model in \eqref{eq:emu} yields (a) a closed-form predictor $\hat{f}_t(\bm{c}_{new}; \Theta_t) = \mathbb{E}\{f_t(\bm{c}_{new}) | \text{Data}\}$ for forces at a new input $\bm{c}_{new}$, and (b) a closed-form uncertainty quantification (UQ) $V_t(\bm{c}_{new}; \Theta_t) = \text{Var}\{f_t(\bm{c}_{new}) | \text{Data}\}$ for this prediction. The full equations for $\hat{f}_t(\bm{c}_{new}; \Theta_t)$ and $V_t(\bm{c}_{new}; \Theta_t)$ can be found in \cite{Sea2013}.

From a Bayesian view, we are interested in the posterior means $\mathbb{E}_{\Theta_t|\text{Data}}[ \hat{f}_t(\bm{c}_{new}; \Theta_t) ]$ and $\mathbb{E}_{\Theta_t|\text{Data}}[ V_t(\bm{c}_{new}; \Theta_t) ]$, since parameters $\Theta_t$ are not known in practice. Using posterior samples $\Theta_t^{(1)}, \cdots ,\Theta_t^{(N)} \sim [\Theta_t|\text{Data}]$, these two quantities can be estimated via:
\begin{equation}
\frac{1}{N} \sum_{m=1}^N \hat{f}_t(\bm{c}_{new}; \Theta_t^{(m)}) \quad \text{and} \quad \frac{1}{N} \sum_{m=1}^N V_t(\bm{c}_{new}; \Theta_t^{(m)}), \quad t = 1, \cdots, T.
\label{eq:emuest}
\end{equation}
The bottleneck is now apparent: every evaluation of $\hat{f}_t$ and $V_t$ requires $\mathcal{O}(D^3)$ work (see \citealp{Sea2013}), meaning the estimators in \eqref{eq:emuest} require $\mathcal{O}(NTD^3)$ work to compute. As total time steps $T$ and design size $D$ grow large, \eqref{eq:emuest} becomes intractable to compute for the full chain $\{\Theta_t^{(1)}, \cdots ,\Theta_t^{(N)}\}_{t=1}^T$, and some reduction of this chain is necessary.

Our set-up is as follows. First, $D=30$ experiments are conducted using a MaxPro design \citep{Jea2015}. Next, for each time step $t$, we obtain $N=50,000$ MCMC samples from the posterior $[\Theta_t|\text{Data}]$, and reduce this down to $n \ll N$ points. Finally, with the new input $\bm{c}_{new}$ chosen as the center of the design region (Figure \ref{fig:mill}, bottom left), prediction and UQ are then performed via \eqref{eq:emuest} using the reduced sample. Four methods are tested: thinning, herding, SPs, and PSPs, with these methods judged on how well they estimate the desired posterior quantities $\mathbb{E}_{\Theta_t|\text{Data}}[ \hat{f}_t(\bm{c}_{new}; \Theta_t) ]$ and $\mathbb{E}_{\Theta_t|\text{Data}}[ V_t(\bm{c}_{new}; \Theta_t) ]$ (true quantities are estimated via a longer MCMC chain with 200,000 samples).

\begin{table}[t]
\centering
\begin{tabular}{c c c c | c c c}
\toprule
 & \multicolumn{3}{c}{\textbf{One-shot} ($n=100$)} & \multicolumn{3}{c}{\textbf{Sequential} ($n=1,000$)}\\
 & Thinning & SPs & PSPs & Thinning & Herding & PSPs\\
 \toprule
 \textit{Prediction} & 0.094 & 0.072 & \cblb{0.069} & 0.073 & 0.080 & \cblb{0.068}\\
 \textit{UQ} & 0.297 & 0.221 & \cblb{0.203} & 0.214 & 1.680 & \cblb{0.201}\\
\hline
\textit{Running time} & 1,648 & 1,746 & 1,926 & 15,972 & 16,876 & 17,392 \\
\toprule
\end{tabular}
\caption{\textup{(Top)} Mean-squared, time-avg. errors for posterior prediction $\mathbb{E}_{\Theta_t|\text{Data}}[ \hat{f}_t(\bm{c}_{new}; \Theta_t) ]$ and UQ $\mathbb{E}_{\Theta_t|\text{Data}}[ V_t(\bm{c}_{new}; \Theta_t) ]$. The method with lowest error is highlighted in \cblb{blue}. \textup{(Bottom)} Running time (in sec.) for reduction, prediction and UQ.}
\label{tbl:emures}
\end{table}

Table \ref{tbl:emures} (top) summarizes the resulting time-averaged errors, split by one-shot and sequential methods. For one-shot methods ($n=100$), PSPs provide the smallest errors for both prediction and UQ, followed closely by SPs, with thinning yielding the largest errors. The poor performance of thinning is not surprising, since it throws away valuable information from the full MCMC chain. For sequential methods, PSPs again offer smaller errors to thinning and herding. As in simulations, herding performs noticeably worse than thinning. From an engineering view, one reason is that not all design inputs (and certainty not all of its interactions) are useful for predicting tangential force. By accounting for this expected sparsity structure in MCMC reduction, the proposed PSPs can yield improved estimation of downstream posterior quantities, as demonstrated here.

Lastly, Table \ref{tbl:emures} (bottom) summarizes the running time of these methods (reduction, prediction and UQ) on a single-core 3.4 Ghz processor. For one-shot (sequential) reduction, thinning is the quickest method, followed closely by SPs (herding) and PSPs, with all methods requiring less than 1 hour (6 hours) of running time. To contrast, the prediction and UQ using the full $N=50,000$ MCMC chain requires $798,600$ seconds ($\approx 9$ days). This is clearly impractical, since a \textit{new} simulation can be performed within the time needed for prediction, thereby \textit{defeating} the purpose of emulation in the first place! Given this need for efficient prediction, PSPs offer the best performance of the reduction methods tested.

\section{Conclusion}
\label{sec:concl}

In this paper, we propose a new method for reducing high-dimensional big data into a representative dataset, called projected support points (PSPs). The key novelty here is the sparsity-inducing (SpIn) kernel, which encourages the preservation of low-dimensional features in high-dimensional data. We first present a theoretical framework for understanding the sparsity conditions needed to lift the curse-of-dimensionality for data reduction, reconciling fundamental ideas from experimental design and Quasi-Monte Carlo (QMC). We then provide practical guidelines on the SpIn kernel specification, and propose two algorithms for efficiently computing PSPs. Finally, we demonstrate the effectiveness of PSPs in simulations, and illustrate its applicability in solving two real-world problems, the first for kernel learning and the second for MCMC reduction.

Looking forward, there are many interesting avenues for future work. One such direction is in speeding up the computation of PSPs for large $n$. Another direction is in exploring an adaptive modification of the PSP methodology, which iteratively incorporates posterior learning on $\boldsymbol{\theta}$ to target active dimensions in data reduction.

\textbf{Supplementary materials}: This paper is accompanied by a supplementary file, containing (a) proofs for technical results and (b) implementation details for algorithms.

{\fontsize{11.5}{14}\selectfont
\spacingset{1.28} 
\bibliography{references}}



\pagebreak

\begin{appendices}
\setcounter{equation}{0}
\renewcommand\theequation{A.\arabic{equation}}

\section{Further details on \texttt{psp.mm} and \texttt{psp.mm.seq}}
\label{app:det}

\subsection{Majorization-minimization}
\label{app:mm}
Here, we give a brief overview of the employed optimization technique \textit{majorization-minimization} (MM), following \cite{Lan2016}. Consider first the definition of a \textit{majorization function}:
\begin{definition}
Let $f:\mathbb{R}^s \rightarrow \mathbb{R}$ be an objective to be minimized. A function $h(\bm{z}|\bm{z}')$ \textit{majorizes} $f(\bm{z})$ at $\bm{z}' \in \mathbb{R}^s$ if $h(\bm{z}'|\bm{z}') = f(\bm{z}')$ and $h(\bm{z}|\bm{z}') \geq f(\bm{z})$ for all $\bm{z} \neq \bm{z}'$.
\end{definition}

\noindent Starting at an initial point $\bm{z}^{[0]} \in\mathbb{R}^s$, MM first minimizes the majorization function $h(\cdot | \bm{z}^{[0]})$ in place of the true objective $f$, then iterates the updates $\bm{z}^{[l+1]} \leftarrow \argmin_{\bm{z}} h(\bm{z}|\bm{z}^{[l]})$ until convergence. The point sequence from this update scheme can be shown to have the \textit{descent property} $f(\bm{z}^{[l+1]}) \leq h(\bm{z}^{[l+1]}|\bm{z}^{[l]}) \leq h(\bm{z}^{[l]}|\bm{z}^{[l]}) = f(\bm{z}^{[l]})$, which ensures the sequence of objective values $( f(\bm{z}^{[l]}))_{l=1}^\infty$ is monotonically decreasing. In this sense, MM guarantees better quality solutions as the number of iterations increases -- a desirable property for optimization. The key to computational efficiency for MM is to ``design'' a surrogate function $g$ which not only majorizes $f$, but also admits an easy-to-compute closed-form minimizer.\\

For the problem at hand, we will establish a quadratic majorizer for the blockwise objective in \eqref{eq:pspoptest}, which then admits an efficient iterative map $\mathcal{M}_i$ for optimization. We begin by showing that the kernel $\gamma_{\boldsymbol{\theta}}$ can be both majorized and minorized by appropriately-chosen paraboloids:

\begin{lemma}
Let $\gamma_{\boldsymbol{\theta}}(\bm{z})$ be the shift-invariant form of kernel $\gamma_{\boldsymbol{\theta}}$ in \eqref{eq:gamma} under the POD weights \eqref{eq:podtheta}. For any $\bm{z}' \in \mathbb{R}^p$, $\gamma_{\boldsymbol{\theta}}(\bm{z})$ is majorized at $\bm{z}'$ by the paraboloid:
\begin{equation}
\bar{Q}_{\boldsymbol{\theta}}(\bm{z}|\bm{z}') := \gamma_{\boldsymbol{\theta}}(\bm{z}') - 2 [\gamma_{\boldsymbol{\theta}}(\bm{z}') \Omega_{\boldsymbol{\theta}} \bm{z}' ]^T (\bm{z} - \bm{z}') + 2 (\bm{z} - \bm{z}')^T\nabla^2_{\boldsymbol{\theta}}(\bm{z}')(\bm{z} - \bm{z}'),
\label{eq:major}
\end{equation}
and minorized at $\bm{z}'$ by the paraboloid:
\begin{equation}
\ubar{Q}_{\boldsymbol{\theta}}(\bm{z}|\bm{z}') := \gamma_{\boldsymbol{\theta}}(\bm{z}') \left[ 1 + \bm{z}' \Omega_{\boldsymbol{\theta}} \bm{z}' \right] - \gamma_{\boldsymbol{\theta}}(\bm{z}') \bm{z}^T \Omega_{\boldsymbol{\theta}} \bm{z},
\label{eq:minor}
\end{equation}
where
\[\Omega_{\boldsymbol{\theta}} = \underset{i=1, \cdots, p}{\textup{diag}} \left\{ \sum_{i \in \bm{u} \subseteq [p]} \Gamma_{|\bm{u}|}\prod_{l\in \bm{u}} \theta_l \right\} \quad \text{and} \quad \nabla^2_{\boldsymbol{\theta}}(\bm{z}):= \frac{4}{e \|\Omega_{\boldsymbol{\theta}} \bm{z}\|_2^2} \left(\Omega_{\boldsymbol{\theta}} \bm{z} \right) \left(\Omega_{\boldsymbol{\theta}} \bm{z} \right)^T.\]
\label{lem:majminexp}
\end{lemma}

Consider now the unbiased objective estimate in \eqref{eq:pspoptest}, which as a function of point $\bm{x}_i$, is proportional to:
\begin{equation}
\frac{1}{nR} \sum_{\substack{j=1 \\ j \neq i}}^n \sum_{r=1}^R \gamma_{\boldsymbol{\theta}_r} (\bm{x}_i,\bm{x}_j) - \frac{1}{N_sR} \sum_{m=1}^{N_s} \sum_{r=1}^R \gamma_{\boldsymbol{\theta}_r}(\bm{x}_i,\bm{y}_m).
\label{eq:pspoptest2}
\end{equation}
\noindent Using the majorizing and minorizing paraboloids $\bar{Q}$ and $\ubar{Q}$ in Lemma \ref{lem:majminexp}, one can then derive a quadratic majorizing function for the blockwise objective in \eqref{eq:pspoptest2}:
\begin{lemma}
For fixed $\mathcal{D}_{-i}$, the blockwise objective \eqref{eq:pspoptest2} is majorized at $\bm{x}' \in \mathbb{R}^p$ by:
\[h_i(\bm{x}|\bm{x}';\mathcal{Y}, \vartheta, \mathcal{D}_{-i}) = \frac{1}{nR}\sum_{\substack{j=1 \\ j \neq i}}^n \sum_{r=1}^R \bar{Q}_{\boldsymbol{\theta}_r}(\bm{x} - \bm{x}_j|\bm{x}' - \bm{x}_j) - \frac{1}{N_sR}\sum_{m=1}^{N_s} \sum_{r=1}^R \ubar{Q}_{\boldsymbol{\theta}_r}(\bm{x} - \bm{y}_m|\bm{x}' - \bm{y}_m),\]
which has the unique closed-form minimizer:
\small
\begin{align}
\begin{split}
\mathcal{M}_i(\bm{x}';\mathcal{Y}, \vartheta, \mathcal{D}_{-i}) &= \left( \frac{2}{N_s R} \sum_{m=1}^{N_s} \sum_{r=1}^R \gamma_{\boldsymbol{\theta}_r}(\bm{x}'-\bm{y}_m) \Omega_{\boldsymbol{\theta}_r} + \frac{4}{nR} \sum_{\substack{j=1\\ j \neq i}}^n \sum_{r=1}^R \nabla_{\boldsymbol{\theta}_r}^2(\bm{x}'-\bm{x}_j) \right)^{-1} \\
& \; \left[ \frac{2}{N_s R}\sum_{m=1}^{N_s} \left( \sum_{r=1}^R \gamma_{\boldsymbol{\theta}_r}(\bm{x}'-\bm{y}_m) \Omega_{\boldsymbol{\theta}_r} \right) \bm{y}_m + \frac{2}{nR} \sum_{\substack{j=1\\ j \neq i}}^n \left( \sum_{r=1}^R \gamma_{\boldsymbol{\theta}_r}(\bm{x}' - \bm{x}_j) \Omega_{\boldsymbol{\theta}_r} \right) (\bm{x}'-\bm{x}_j) \right. \\
& \; \quad  \left. + \frac{4}{nR} \sum_{\substack{j=1\\ j \neq i}}^n \sum_{r=1}^R  \nabla_{\boldsymbol{\theta}_r}^2(\bm{x}'-\bm{x}_j) \bm{x}'\right].
\label{eq:closedform2}
\end{split}
\end{align}
\normalsize
\label{lem:majmin}
\end{lemma}

\noindent From \eqref{eq:closedform2}, one can show that the running time for map $\mathcal{M}_i$ is $\mathcal{O}(np)$, assuming subsample sizes $N$ and $R$ are independent of $n$ and $p$.

\subsection{Recursive computation of POD weights}
\label{app:rec}

A key computational problem for both the one-shot and sequential methods \texttt{psp.mm} and \texttt{psp.mm.seq} is the evaluation of the diagonal matrix $\Omega_{\boldsymbol{\theta}}$ in Lemma \ref{lem:majminexp}, which is required for computing the iterative map $\mathcal{M}_i$. Addressing this is particularly important for high-dimensions, because a brute-force evaluation of each entry in $\Omega_{\boldsymbol{\theta}}$ requires $\mathcal{O}(2^p)$ work -- this is infeasible even for moderate choices of $p$. Similar to the recursive component-by-component construction of POD-weighted shifted lattice rules (see Section 5.6 of \citealp{Dea2013}), the following theorem gives a recursive method for efficiently computing $\Omega_{\boldsymbol{\theta}}$:

\begin{theorem}
The $l$-th diagonal of $\Omega_{\boldsymbol{\theta}}$ can be computed as $\Omega_{\boldsymbol{\theta},ll} = \theta_l \sum_{k=1}^p \Gamma_{k}^{(\theta)} r_{p,k-1}^{(-l)}$. For each $l = 1, \cdots, p$, $r_{p,k-1}^{(-l)}$ can be computed recursively as:
\begin{equation}
r_{s,k}^{(-l)} = \theta_{s} r_{s-1,k-1}^{(-l)} + r_{s-1,k}^{(-l)}, \quad s \in [p] \setminus \{l\}, \quad r_{l,k}^{(-l)} = r_{l-1,k}^{(-l)},
\label{eq:rec}
\end{equation}
with initial values $r_{s,0}^{(-l)} = 1$ and $r_{s,k}^{(-l)} = 0$, $k>s$.
\label{thm:podcomp}
\end{theorem}
\noindent The appeal of the recursive scheme in Theorem \ref{thm:podcomp} is that each entry in $\Omega_{\boldsymbol{\theta}}$ can be computed in $\mathcal{O}(p^2)$ work, which is much faster than the $\mathcal{O}(2^p)$ work in a brute-force evaluation. For truncated order weights, i.e., $\Gamma_{|\bm{u}|}^{(\theta)} = 0, |\bm{u}| > K$ for some $K < p$, this running time can be further reduced to $\mathcal{O}(Kp)$. Even in the untruncated setting, one can perform a manual truncation for large $|\bm{u}|$ without sacrificing much accuracy in practice.


\section{Proofs of technical results}
\label{app:proofs}

\subsection{Proof of Theorem \ref{thm:rkhs}}
We require an important lemma to prove this theorem:
\begin{lemma}
\citep{Aro1950}
\label{lem:aro}
Suppose $\mathcal{H}$ is a separable Hilbert space of functions on $\mathcal{X}$ with orthonormal basis $\{\phi_k(\bm{x})\}_{k=0}^\infty$. Then $\mathcal{H}$ is a RKHS if and only if $\sum_{k=0}^\infty |\phi_k(\bm{x})|^2 < \infty$ for any $\bm{x} \in \mathcal{X}$, with unique kernel given by $k(\bm{x},\bm{y})=\sum_{k=0}^\infty \phi_k(\bm{x}) \phi_k(\bm{y})$.
\end{lemma}

\begin{proof}(Theorem \ref{thm:rkhs})
We adopt a similar approach as Minh (2010) to derive the RKHS for $\gamma_{\boldsymbol{\theta}}$. Note that:
\begin{align}
\begin{split}
\gamma_{\boldsymbol{\theta}}(\bm{x},\bm{y}) &= \exp(-\|\bm{x}-\bm{y}\|_{\boldsymbol{\theta}}^2)\\
& = \exp(-\|\bm{x}\|^2_{\boldsymbol{\theta}}) \exp(-\|\bm{y}\|^2_{\boldsymbol{\theta}}) \exp( 2 \langle \bm{x},\bm{y} \rangle_{\boldsymbol{\theta}} )\\
& = \exp(-\|\bm{x}\|^2_{\boldsymbol{\theta}}) \exp(-\|\bm{y}\|^2_{\boldsymbol{\theta}}) \sum_{k=0}^\infty \frac{2^k}{k!} \sum_{|{\boldsymbol{\alpha}}|=k} C_{\boldsymbol{\alpha}}^k \bm{x}^{\boldsymbol{\alpha}} \bm{y}^{\boldsymbol{\alpha}} \boldsymbol{\theta}^{\boldsymbol{\alpha}},
\end{split}
\label{eq:expn}
\end{align}
where the last step follows by the series expansion:
\[\exp(2 \langle \bm{x},\bm{y} \rangle_{\boldsymbol{\theta}}) = \sum_{k=0}^\infty \frac{2^k \langle \bm{x}, \bm{y} \rangle_{\boldsymbol{\theta}}^k}{k!} = \sum_{k=0}^\infty \frac{2^k}{k!} \sum_{|{\boldsymbol{\alpha}}|=k} C_{\boldsymbol{\alpha}}^k \bm{x}^{\boldsymbol{\alpha}} \bm{y}^{\boldsymbol{\alpha}} \boldsymbol{\theta}^{\boldsymbol{\alpha}}.\]

Now, assume $\mathcal{H}_{\gamma_{\boldsymbol{\theta}}}$ is the space in \eqref{eq:rkhs} with inner product \eqref{eq:inner}. The completeness of $\mathcal{H}_{\gamma_{\boldsymbol{\theta}}}$ can be shown using a similar argument in Minh (2010), so $(\mathcal{H}_{\gamma_{\boldsymbol{\theta}}}, \langle \cdot, \cdot \rangle_{\gamma_{\boldsymbol{\theta}}})$ is a valid Hilbert space. Define the basis $\phi_{\boldsymbol{\alpha}}(\bm{x}) = \sqrt{{2^k C_{\boldsymbol{\alpha}}^k \boldsymbol{\theta}^{\boldsymbol{\alpha}} } / {|\boldsymbol{\alpha}|!}} \exp(-\|\bm{x}\|^2_{\boldsymbol{\theta}}) \bm{x}^{\boldsymbol{\alpha}}$, $|{\boldsymbol{\alpha}}| \in \mathbb{N}_0$, and note that (a) this basis is orthonormal under the inner product in \eqref{eq:inner}, and (b) $\text{span}\{\phi_{{\boldsymbol{\alpha}}}(\bm{x})\}=\mathcal{H}_{\gamma_{\boldsymbol{\theta}}}$, which shows $\mathcal{H}_{\gamma_{\boldsymbol{\theta}}}$ is separable. Moreover, because $\sum_{k=0}^\infty \sum_{|{\boldsymbol{\alpha}}|=k} \phi_{{\boldsymbol{\alpha}}}^2(\bm{x}) < \infty$ and:
\[\sum_{k=0}^\infty \sum_{|{\boldsymbol{\alpha}}|=k} \phi_{{\boldsymbol{\alpha}}}(\bm{x}) \phi_{{\boldsymbol{\alpha}}}(\bm{y}) = \sum_{k=0}^\infty \sum_{|{\boldsymbol{\alpha}}|=k} {\frac{2^k C_{\boldsymbol{\alpha}}^k \boldsymbol{\theta}^{\boldsymbol{\alpha}} }{k!}} \exp(-\|\bm{x}\|^2_{\boldsymbol{\theta}}) \exp(-\|\bm{y}\|^2_{\boldsymbol{\theta}}) \bm{x}^{\boldsymbol{\alpha}} \bm{y}^{\boldsymbol{\alpha}} = \gamma_{\boldsymbol{\theta}}(\bm{x},\bm{y}),\]
it follows by Lemma \ref{lem:aro} that $(\mathcal{H}_{\gamma_{\boldsymbol{\theta}}}, \langle \cdot, \cdot \rangle_{\gamma_{\boldsymbol{\theta}}})$ is the RKHS corresponding to kernel $\gamma_{\boldsymbol{\theta}}$.
\end{proof}


\subsection{Proof of Theorem \ref{thm:dim}}
To prove this theorem, we require a lemma:
\begin{lemma}
For fixed $p$ and $\boldsymbol{\alpha} = (\alpha_1, \cdots, \alpha_p)$, $\alpha_l \in \mathbb{Z}_+$, $\lim_{k \rightarrow \infty} \sum_{|\boldsymbol{\alpha}| = k} 1/C^k_{\boldsymbol{\alpha}} = p$.
\label{lem:combn}
\end{lemma}
\begin{proof}(Lemma \ref{lem:combn})
Fix $p \in \mathbb{Z}_+$, and consider the following decomposition for sufficiently large $k \in \mathbb{Z}_+$:
\[\sum_{|\boldsymbol{\alpha}| = k} \frac{1}{C^k_{\boldsymbol{\alpha}}} = \sum_{|\boldsymbol{\alpha}| = k, \exists \alpha_l = k} \frac{1}{C^k_{\boldsymbol{\alpha}}} + \sum_{|\boldsymbol{\alpha}| = k, \exists \alpha_l = k-1} \frac{1}{C^k_{\boldsymbol{\alpha}}} + \cdots + \sum_{|\boldsymbol{\alpha}| = k, \exists \alpha_l = k-p+1} \frac{1}{C^k_{\boldsymbol{\alpha}}} + \sum_{|\boldsymbol{\alpha}| = k, \alpha_l \leq k-p} \frac{1}{C^k_{\boldsymbol{\alpha}}}.\]
For the first sum, it is easy to see that $\sum_{|\boldsymbol{\alpha}| = k, \exists \alpha_l = k} 1/C^k_{\boldsymbol{\alpha}} = p$, because there are $p$ terms in this sum, with each term equal to 1. For the second sum, one can similarly show that $\sum_{|\boldsymbol{\alpha}| = k, \exists \alpha_l = k-1} 1/C^k_{\boldsymbol{\alpha}} = \mathcal{O}(p^2/k)$, because there are $\mathcal{O}(p^2)$ terms in this sum, with each term bounded above by $1/k$. Extending the same argument for remaining terms, the above decomposition can be rewritten as:
\[ \sum_{|\boldsymbol{\alpha}| = k} \frac{1}{C^k_{\boldsymbol{\alpha}}} = p + \mathcal{O}\left(\frac{p^2}{k}\right) + \cdots + \mathcal{O}\left(\frac{p^{p+1}}{k(k-1) \cdots (k-p+1)}\right) + \sum_{|\boldsymbol{\alpha}| = k, \alpha_l \leq k-p} \frac{1}{C^k_{\boldsymbol{\alpha}}}.\]
Consider now the last sum $\sum_{|\boldsymbol{\alpha}| = k, \alpha_l \leq k-p} 1/C^k_{\boldsymbol{\alpha}}$. Note that $|\{ \boldsymbol{\alpha}: \sum_l \alpha_l = k\}| = {{k-1} \choose {p-1}}$ (this is the number of ways to put $k$ balls in $p$ containers), so there are at most ${{k-1} \choose {p-1}}$ terms in this term. Moreover, $1/C^k_{\boldsymbol{\alpha}} \leq p!/(k (k-1) \cdots (k-p+1))$ whenever $|\boldsymbol{\alpha}| = k$, $\alpha_l \leq k-p$. Combining these two facts, we get $\sum_{|\boldsymbol{\alpha}| = k, \alpha_l < k-p} 1/C^k_{\boldsymbol{\alpha}} \leq p/k$. Hence:
\[ \lim_{k\rightarrow \infty}\sum_{|\boldsymbol{\alpha}| = k} \frac{1}{C^k_{\boldsymbol{\alpha}}} = \lim_{k \rightarrow \infty} \left\{p + \mathcal{O}\left(\frac{p^2}{k}\right) + \cdots + \mathcal{O}\left(\frac{p^{p+1}}{k(k-1) \cdots (k-p+1)}\right)  + \frac{p}{k} \right\} = p.\]
\end{proof}

\begin{proof}(Theorem \ref{thm:dim})
Take first Lemma \ref{lem:kh}, which can be rewritten as the upper bound:
\begin{equation}
I(g;F,F_n) \leq \|g\|_{\gamma_{\boldsymbol{\theta}}} D_{\gamma_{\boldsymbol{\theta}}}(F,F_n).
\label{eq:kh2}
\end{equation}
This bound has two parts: the norm $\|g\|_{\gamma_{\boldsymbol{\theta}}}$, and the discrepancy $D_{\gamma_{\boldsymbol{\theta}}}(F,F_n)$. We will consider each part separately below.

Consider first the discrepancy $D_{\gamma_{\boldsymbol{\theta}}}(F,F_n)$. Let $\tilde{F}_n$ be the e.d.f. of the random point set $\{\bm{x}_i\}_{i=1}^n$, where $(\bm{x}_i)_{i=1}^\infty \distas{i.i.d.} F$. The expected discrepancy for this point sequence becomes:
\begin{align*}
\mathbb{E}[D_{\gamma_{\boldsymbol{\theta}}}^2(F,F_n)] &= \mathbb{E}_{\{\bm{x}_i\}}\left[ \mathbb{E}\{\gamma_{\boldsymbol{\theta}}(\bm{Y}, \bm{Y}')\} - \frac{2}{n} \sum_{i=1}^n \mathbb{E}\{\gamma_{\boldsymbol{\theta}}(\bm{x}_i, \bm{Y})\} + \frac{1}{n^2} \sum_{i=1}^n \sum_{j=1}^n \gamma_{\boldsymbol{\theta}}(\bm{x}_i, \bm{x}_j)\right]\\
&= \mathbb{E}\left[ \gamma_{\boldsymbol{\theta}}(\bm{Y}, \bm{Y}') - \frac{2}{n} \sum_{i=1}^n \mathbb{E}_{\{\bm{x}_i\}}\{\gamma_{\boldsymbol{\theta}}(\bm{x}_i, \bm{Y})\} + \frac{1}{n^2} \sum_{i=1}^n \sum_{j=1}^n \mathbb{E}_{\{\bm{x}_i\}}\{\gamma_{\boldsymbol{\theta}}(\bm{x}_i, \bm{x}_j)\}\right]\\
&= \mathbb{E} \{ \gamma_{\boldsymbol{\theta}}(\bm{Y}, \bm{Y}') \} - 2 \mathbb{E}\{\gamma_{\boldsymbol{\theta}}(\bm{Y}, \bm{Y}')\} + \frac{1}{n^2} \left[ n(n-1) \mathbb{E}\{\gamma_{\boldsymbol{\theta}}(\bm{Y}, \bm{Y}')\} + n \mathbb{E}\{\gamma_{\boldsymbol{\theta}}(\bm{Y}, \bm{Y}) \} \right]\\
&= \frac{1}{n} \left[ 1 -  \mathbb{E}\{\gamma_{\boldsymbol{\theta}}(\bm{Y}, \bm{Y}')\} \right] \leq \frac{1}{n}.
\end{align*}
Because PSPs are defined as the \textit{minimizer} of $D_{\gamma_{\boldsymbol{\theta}}}(F,F_n)$, it follows by the above averaging argument that $D_{\gamma_{\boldsymbol{\theta}}}(F,F_n) \leq 1/\sqrt{n}$.

Consider next the second term $\|g\|_{\gamma_{\boldsymbol{\theta}}}$. By Theorem \ref{thm:rkhs}, we have:
\begin{align*}
\|g\|_{\gamma_{\boldsymbol{\theta}}}^2 = \sum_{k=0}^\infty \frac{k!}{2^k} \sum_{|{\boldsymbol{\alpha}}|=k} \frac{w_{\boldsymbol{\alpha}}^2}{C_{\boldsymbol{\alpha}}^k {\boldsymbol{\theta}}^{\boldsymbol{\alpha}}} &= \sum_{k=0}^\infty \frac{k!}{2^k} \sum_{|{\boldsymbol{\alpha}}|=k} \left(\frac{1}{(C_{\boldsymbol{\alpha}}^k)^{3/2}}\right) \left(\frac{\sqrt{C_{\boldsymbol{\alpha}}^k}w_{\boldsymbol{\alpha}}^2} { {\boldsymbol{\theta}}^{\boldsymbol{\alpha}}}\right)\\
&\leq \sum_{k=0}^\infty \frac{k!}{2^k} \sqrt{\sum_{|{\boldsymbol{\alpha}}|=k} \frac{1}{(C_{\boldsymbol{\alpha}}^k)^3} } \sqrt{ \sum_{|{\boldsymbol{\alpha}}|=k} \frac{C_{\boldsymbol{\alpha}}^k w_{\boldsymbol{\alpha}}^4} { {\boldsymbol{\theta}}^{2\boldsymbol{\alpha}}} } \tag{Cauchy-Schwarz}\\
&= \sum_{k=0}^\infty \frac{C}{\sqrt{p}2^k} \sqrt{\sum_{|{\boldsymbol{\alpha}}|=k} \frac{1}{(C_{\boldsymbol{\alpha}}^k)^3} } \sqrt{ \sum_{|{\boldsymbol{\alpha}}|=k} C_{\boldsymbol{\alpha}}^k \prod_{l=1}^p \left( \frac{w_l^4} { \theta_l^2}\right)^{\alpha_l} } \tag{POD form of $w_{\boldsymbol{\alpha}}$ and $\Gamma_{|\boldsymbol{\alpha}|}^{(w)} \leq C/\{p^{-1/4}(|\boldsymbol{\alpha}|!)^{-1/2}\}$} \\
&\leq \sum_{k=0}^\infty \frac{C}{\sqrt{p}2^k} \sqrt{\sum_{|{\boldsymbol{\alpha}}|=k} \frac{1}{(C_{\boldsymbol{\alpha}}^k)} } \sqrt{ \sum_{|{\boldsymbol{\alpha}}|=k} C_{\boldsymbol{\alpha}}^k \prod_{l=1}^p \left(\frac{w_l^4} { \theta_l^2}\right)^{\alpha_l} } \\
&\leq \sum_{k=0}^\infty \frac{C}{\sqrt{p} 2^k} \sqrt{\sum_{|{\boldsymbol{\alpha}}|=k} \frac{1}{(C_{\boldsymbol{\alpha}}^k)} } \left(\sqrt{\sum_{l=1}^p \frac{w_l^4} { \theta_l^2}} \right)^{k} \tag{Binomial theorem}.
\end{align*}
Taking the limit as $k \rightarrow \infty$, Lemma \ref{lem:combn} gives $\sqrt{\sum_{|{\boldsymbol{\alpha}}|=k} {1}/{(C_{\boldsymbol{\alpha}}^k)} } \rightarrow \sqrt{p}$. Finally, if $\sum_{l=1}^\infty {w_l^4}/{ \theta_l^2} < 4$, the above series converges to a constant independent of $p$, as desired. Combining this with the upper bound in \eqref{eq:kh2}, the proof is complete.
\end{proof}

\subsection{Proof of Theorem \ref{thm:pspconv}}
\begin{proof}
This follows by a direct extension of Theorems 4 and 5 in \cite{MJ2016b}.
\end{proof}

\subsection{Proof of Proposition \ref{prop:closed}}
\begin{proof}
Let $\pi$ be the i.i.d. Gamma priors in \eqref{eq:prior}. The SpIn kernel $\mathbb{E}_{\boldsymbol{\theta} \sim \pi}\{\gamma_{\boldsymbol{\theta}} (\bm{x},\bm{y})\}$ can be rewritten as follows:
\begin{align*}
\begin{split}
\mathbb{E}_{\boldsymbol{\theta} \sim \pi}\{\gamma_{\boldsymbol{\theta}}(\bm{x},\bm{y})\} &= \mathbb{E}_{\boldsymbol{\theta} \sim \pi}\left[ \exp\left\{ -\sum_{l=1}^p \theta_l (x_l - y_l)^2 \right\}\right]\\
&= \prod_{l=1}^p \left[ \int_0^\infty \exp\{ - \theta_l (x_{l}-y_l)^2\} \cdot \left\{\frac{\lambda^\nu}{\Gamma(\nu)} \theta_l^{\nu-1} \exp(-\lambda \theta_l)\right\} \; d\theta_l \right]\\
&= \prod_{l=1}^p \left\{ \frac{\lambda}{(x_l - y_l)^2 + \lambda} \right\}^\nu,
\end{split}
\end{align*}
which completes the proof.
\end{proof}

\subsection{Proof of Lemma \ref{lem:majminexp}}
\begin{proof}
First consider the majorizing paraboloid $\bar{Q}$ in \eqref{eq:major}. It is easy to show that:
\[\nabla_{\bm{z}} \gamma_{\boldsymbol{\theta}}(\bm{z}) = - 2 \gamma_{\boldsymbol{\theta}}(\bm{z}) \Omega_{\boldsymbol{\theta}} \bm{z} \quad \text{and} \quad \nabla^2_{\bm{z}} \gamma_{\boldsymbol{\theta}}(\bm{z}) = 2 \gamma_{\boldsymbol{\theta}}(\bm{z}) \left[ 2 \Omega_{\boldsymbol{\theta}}\bm{z} (\Omega_{\boldsymbol{\theta}} \bm{z})^T - \Omega_{\boldsymbol{\theta}} \right].\]
Note that, for any $\bm{z} \in \mathbb{R}^p$:
\begin{align*}
\nabla^2 \gamma_{\boldsymbol{\theta}}(\bm{z}) \preceq 4 \gamma_{\boldsymbol{\theta}}(\bm{z}) (\Omega_{\boldsymbol{\theta}} \bm{z})(\Omega_{\boldsymbol{\theta}} \bm{z})^T &\preceq 4 \gamma_{\boldsymbol{\theta}}(\bm{z}) \|\Omega_{\boldsymbol{\theta}} \bm{z}\|_2^2 \bm{I}_p\\
& \preceq 4 \exp\left\{ - \sum_{l=1}^p \Omega_{\boldsymbol{\theta},ll} \|\bm{z}_l\|_2^2 \right\} \left( \sum_{l=1}^p \Omega_{\boldsymbol{\theta},ll} \|\bm{z}_l\|_2^2 \right) \left( \max_l \Omega_{\boldsymbol{\theta},ll} \right) \bm{I}_p\\
& \preceq \frac{4}{e} \left( \max_l \Omega_{\boldsymbol{\theta},ll} \right) \bm{I}_p = 4\Delta_{\boldsymbol{\theta}} \tag{$\displaystyle \min_z \exp\{- z^2\}z^2 = \frac{1}{ e}$}.
\end{align*}
Using a second-order Taylor expansion of $\gamma_{\boldsymbol{\theta}}(\bm{z})$ at $\bm{z} = \bm{z}'$, the following must hold for some $\boldsymbol{\xi} = (1-t)\bm{z} + t\bm{z}'$ with $t \in [0,1]$:
\[\gamma_{\boldsymbol{\theta}}(\bm{z}) = \gamma_{\boldsymbol{\theta}}(\bm{z}') - 2 [\gamma_{\boldsymbol{\theta}}(\bm{z}') \Omega_{\boldsymbol{\theta}} \bm{z}']^T (\bm{z} - \bm{z}') + \frac{1}{2} (\bm{z} - \bm{z}')^T [\nabla^2 \gamma_{\boldsymbol{\theta}}(\boldsymbol{\xi})] (\bm{z} - \bm{z}') \leq \bar{Q}(\bm{z}|\bm{z}').\]
By definition, $\bar{Q}(\bm{z}|\bm{z}')$ majorizes $\gamma_{\boldsymbol{\theta}}(\bm{z})$ at $\bm{z} = \bm{z}'$.

Next, consider the minorizing paraboloid $\ubar{Q}$ in \eqref{eq:minor}. Note that $\exp(t) \geq (1-t')\exp(t') + t\exp(t')$ by convexity. Hence:
\small
\begin{align*}
\gamma_{\boldsymbol{\theta}}(\bm{z}) &\geq \gamma_{\boldsymbol{\theta}}(\bm{z}') \left[ 1 + \sum_{\varnothing \neq \bm{u} \subseteq [p]} \theta_{\bm{u}} \|\bm{z}'_{\bm{u}}\|_2^2 \right] - \gamma_{\boldsymbol{\theta}}(\bm{z}') \sum_{\varnothing \neq \bm{u} \subseteq [p]} \theta_{\bm{u}} \|\bm{z}_{\bm{u}}\|_2^2 = \gamma_{\boldsymbol{\theta}}(\bm{z}') \left[ 1 + \bm{z}' \Omega_{\boldsymbol{\theta}} \bm{z}' \right] - \gamma_{\boldsymbol{\theta}}(\bm{z}') \bm{z}^T \Omega_{\boldsymbol{\theta}} \bm{z},
\end{align*}
\normalsize
which completes the proof.
\end{proof}

\subsection{Proof of Lemma \ref{lem:majmin}}
\begin{proof}
The majorization claim follows directly from Lemma \ref{lem:majminexp}, and the closed-form minimizer can be obtained by setting the gradient of $h_i$ to zero and solving for $\bm{x}$.
\end{proof}

\subsection{Proof of Theorem \ref{thm:algconv}}
\begin{proof}
Given the result in Lemma \ref{lem:majmin}, this theorem can be proven by Prop. 3.4 of Mairal (2013), under certain regularity conditions. These conditions are satisfied by the convexity and compactness of $\mathcal{X}$, and the differentiability of $\gamma_{\boldsymbol{\theta}}(\cdot)$.
\end{proof}

\subsection{Proof of Theorem \ref{thm:podcomp}}
\begin{proof}
Starting from the $i$-th entry of the diagonal of $\Omega_{\boldsymbol{\theta}}$, $i = 1, \cdots, p$, we get:
\begin{align*}
\Omega_{\boldsymbol{\theta},ii} &= \sum_{i \in \bm{u} \subseteq [p]} \Gamma_{|\bm{u}|}^{(\theta)} \prod_{l\in \bm{u}} \theta_l\\
&= \sum_{k=1}^p \sum_{i \in \bm{u} \subseteq [p], |\bm{u}| = k} \Gamma_{|\bm{u}|}^{(\theta)} \prod_{l\in \bm{u}} \theta_l\\
&= \theta_i \sum_{k=1}^p \Gamma_{k}^{(\theta)} \sum_{\bm{u} \subseteq [p] \setminus \{i\}, |\bm{u}| = k-1} \prod_{l\in \bm{u}}\theta_l\\
&= \theta_i \sum_{k=1}^p \Gamma_{k}^{(\theta)} r_{p,k-1}^{(-i)},
\end{align*}
where $\displaystyle r_{s,k}^{(-i)} = \sum_{\bm{u} \subseteq [s] \setminus \{i\}, |\bm{u}| = k} \prod_{l\in \bm{u}}\theta_l$ for $s = 0, \cdots, p$. For $s > 0$, $s \neq i$, note that:
\[r_{s,k}^{(-i)} = \sum_{s \in \bm{u} \subseteq [s] \setminus \{i\}, |\bm{u}| = k} \prod_{l\in \bm{u}}\theta_l + \sum_{s \notin \bm{u} \subseteq [s] \setminus \{i\}, |\bm{u}| = k} \prod_{l\in \bm{u}}\theta_l = \theta_{s} r_{s-1,k-1}^{(-i)} + r_{s-1,k}^{(-i)},\]
with initial values $r_{s,0}^{(-i)} = 1$ and $r_{s,k}^{(-i)} = 0$ for $k > s$. This proves the correctedness of the recursive procedure.
\end{proof}

\end{appendices}

\end{document}